\documentclass[journal,onecolumn,twoside]{IEEEtran}
\normalsize
\ifCLASSINFOpdf
\else
\fi
\usepackage{setspace}
\usepackage{color}
\usepackage{cite}
\usepackage{amsmath}
\usepackage{amsfonts}
\usepackage{amssymb}
\usepackage{amsthm}
\usepackage{tikz}
\usepackage{cases}
\usepackage{bm}
\usepackage{graphicx}
    \graphicspath{{../}}
    \DeclareGraphicsExtensions{.pdf}
\usepackage[caption=true,font=footnotesize]{subfig}
\usepackage{multirow}
\usepackage{makecell}
\usepackage{mathdots}
\usepackage{booktabs}
\usepackage{url}
\usepackage{bm}
\usepackage{xtab}
\usepackage{pifont}
\usepackage{array}
\usepackage{algorithm}
\usepackage{algorithmic}
\usepackage{enumerate}
\usetikzlibrary{arrows}
\usepackage{caption}
\usepackage{lineno} 
\usepackage{mathrsfs}
\usepackage{tikz-cd}
\usepackage{mathtools}
\usepackage{makecell}
\usepackage{booktabs}
\usepackage{array}
\usepackage{enumitem}
\usepackage[table]{xcolor}
\usepackage{graphicx}
\definecolor{readcolor}{RGB}{255,235,156}
\newcommand{\readcell}[1]{\cellcolor{readcolor}#1}
\usepackage{adjustbox}

\newcommand{\tinytabstyle}{%
  \tiny
  \setlength{\tabcolsep}{1.6pt}%
  \renewcommand{\arraystretch}{0.50}%
  \setlength{\fboxsep}{0.3pt}%
}
\allowdisplaybreaks[4]
\usepackage[colorlinks,
            linkcolor=blue,
            anchorcolor=blue,
            citecolor=blue,
            urlcolor=black]{hyperref}
\theoremstyle{plain}

\newtheorem{theorem}{Theorem}[section]
\newtheorem{lemma}{Lemma}[section]
\newtheorem{proposition}{Proposition} [section]
\newtheorem{definition}{Definition} [section]
\newtheorem{corollary}{Corollary} [section]
\theoremstyle{definition}

\newtheorem{remark}{Remark}[section]

\begin{document}	
\title{Tight Bandwidth Lower Bounds and Optimal Constructions of Locally Repairable Convertible Codes in the Global Split Regime}
\author{Haoming Shi,~\IEEEmembership{}
        Weijun~Fang~\IEEEmembership{}
        
\IEEEcompsocitemizethanks{\IEEEcompsocthanksitem Haoming Shi and Weijun Fang are with the State Key Laboratory of Cryptography and Digital Economy Security, Shandong University, Qingdao, 266237, China, the Key Laboratory of Cryptologic Technology and Information Security, Ministry of Education, Shandong University, Qingdao, 266237, China and the School of Cyber Science and Technology, Shandong University, Qingdao, 266237, China (emails: 202421328@mail.sdu.edu.cn, fwj@sdu.edu.cn).
}
\thanks{This work was supported in part by the National Key Research and Development Program of China under Grant Nos. 2022YFA1004900 and 2021YFA1001000, the National Natural Science Foundation of China under Grant No. 62571301, and the Shandong Provincial Natural Science Foundation under Grant No. ZR2025QA05. {\it (Corresponding Author: Weijun Fang)}.}
}

\maketitle
\begin{abstract}
Erasure coding is a key technique for providing fault tolerance in modern distributed storage systems. In practice, as storage systems evolve, the parameters of the deployed erasure code may need to be adjusted to accommodate changes in storage scale, reliability requirements, and disk failure rates. Such adaptation is achieved through code conversion, which transforms data encoded by an initial code into data encoded by a final code. Convertible codes are designed to carry out this transformation efficiently while preserving desirable code properties.

In this work, we study code conversion between systematic optimal-distance locally repairable codes (LRCs) in the global split regime, using read bandwidth as the conversion-efficiency metric. Specifically, we focus on the parameter range $g^I,g^F \leq r$, where the numbers of initial and final global parity nodes are at most the local information dimension $r$. Over this entire parameter range, we derive lower bounds on the read bandwidth of stable optimal-distance locally repairable convertible codes (LRCCs) via an information-theoretic approach, without imposing any linearity assumption on the initial codes, the final codes, or the conversion procedure. We then develop constructions based on MDS array codes with prescribed repair or alignment properties. Depending on the relative sizes of \(g^I\) and \(g^F\), we handle the construction separately in the three cases \(g^F=g^I\), \(g^F>g^I\), and \(g^F<g^I\), and show that each attains the corresponding lower bound. This yields a complete characterization of the optimal read bandwidth for stable optimal-distance LRCCs over the entire parameter range \(g^I,g^F\le r\).

\end{abstract}


\IEEEpeerreviewmaketitle

\section{Introduction}
\label{se1}
In modern distributed storage systems, large volumes of data are distributed across many nodes, where node failures and disk errors are unavoidable. Erasure codes are therefore widely used to provide fault tolerance with lower storage overhead than replication~\cite{ghemawat2003google,huang2012erasure,sathiamoorthy2013xoring}. Among them, maximum distance separable codes, or MDS codes, achieve the optimal tradeoff between storage efficiency and fault tolerance: an $(n,k)$ MDS code can recover the original data from any $k$ out of $n$ encoded symbols, or equivalently, can tolerate any $n-k$ erasures\cite{singleton1964maximum}. However, repairing a failed node in an MDS-coded storage system often requires downloading data from many surviving nodes, leading to high repair bandwidth~\cite{rashmi2013solution,rashmi2014hitchhiker}. Locally repairable codes (LRCs)~\cite{gopalan2012locality} address this issue by adding local parity symbols to small local groups, so that a failed symbol can be recovered by accessing only a small number of other symbols in its group. Thus, while MDS codes maximize erasure tolerance for a given amount of redundancy, LRCs provide a more practical balance among reliability, storage overhead, and repair efficiency in distributed storage systems~\cite{balaji2018erasure,huang2013pyramid}.

In practical storage systems, the failure rates of storage nodes may vary over time~\cite{kadekodi2019cluster}. This temporal variation is often captured by a bathtub curve, with an early-failure period, a relatively stable useful-life period, and a wear-out period. As the underlying failure rates evolve, a fixed set of coding parameters may no longer be optimal for balancing reliability, storage overhead, and repair efficiency. To exploit such time-varying characteristics, prior work\cite{kim2024morph} showed that dynamically adapting the code rate can lead to substantial savings in both storage and operational costs. Adjusting the code rate requires changing the parameters of the underlying code. This procedure was formalized as code conversion\cite{maturana2022convertible}, where data encoded by an initial code are transformed into data encoded by a final
code. A naive approach to performing such a conversion is to fully decode the existing data and then re-encode it using the final code, incurring substantial read and write costs. Convertible codes are designed to carry out this process more efficiently while preserving the desired properties of the initial and final codes.

In this paper, we study code conversion for optimal-distance $(k,g,r,\ell)$ locally repairable codes\cite{maturana2023locally} in the
global split regime. Specifically, we consider stable conversions from an optimal-distance initial $(k^I,g^I,r,\ell)$-LRC to an optimal-distance final $(k^F,g^F,r,\ell)$-LRC, where one initial codeword is converted into $\lambda^F$ final codewords, with $k^I=\lambda^F k^F$ and the locality parameters $r$ and $\ell$ preserved. We focus on the parameter range $g^I,g^F\le r$, where the numbers of initial and final global parity nodes do not exceed the local information dimension. Since stable conversions maximize the number of unchanged nodes, the write cost is fixed in this setting, and the central question becomes how much data must be read from the initial codeword in order to generate the final codewords. Our goal is to establish read-bandwidth lower bounds for such conversions and to determine whether these bounds can be attained by matching constructions.

\subsection{Related Work}

The problem of code conversion was first introduced by Maturana and
Rashmi~\cite{maturana2022convertible} as a framework for efficiently updating the coding parameters of already encoded data in distributed storage systems. Their early work focused on conversions between MDS codes, with particular emphasis on the merge regime, where multiple initial codewords are combined into a single final codeword. Under the access cost metric, they established fundamental lower bounds and provided optimal constructions attaining these bounds.

Subsequent works refined the study of convertible codes in several directions. 
First, the conversion cost has been evaluated under different metrics, including access cost, which counts the number of accessed nodes during conversion, and bandwidth cost, which measures the total amount of data read from and written to storage nodes. 
Second, code conversion has been studied in several regimes, including the merge regime, the split regime, and more general regimes between arbitrary initial and final parameters.
For MDS convertible codes, several works \cite{maturana2022bandwidth,maturana2020access,maturana2023bandwidth,maturana2024code,chopra2024low,singhvi2025tight,ge2024mds,ramkumar2026mds,wang2026lower} have established lower bounds and optimal constructions under these different regimes and cost metrics.
In particular, an information-theoretic approach has recently been used to derive lower bounds on the bandwidth cost of MDS convertible codes in the split regime, without assuming linearity of the codes or the conversion procedure, or uniform downloads from the initial symbols~\cite{singhvi2025tight}.

Before the formal notion of convertible codes was introduced, related ideas had already appeared in the context of LRC-based storage systems. Xia et
al.~\cite{xia2015tale} proposed a storage architecture using two different LRCs to encode frequently accessed data and less frequently accessed data separately. Their scheme includes a mechanism for converting between the two LRCs. This mechanism is designed to reduce the number of nodes accessed during conversion when the access frequency of the data changes.
Later, Maturana and Rashmi~\cite{maturana2023locally} formally initiated the study of LRC conversion within the framework of convertible codes. Under the bandwidth cost metric, they proposed constructions for LRC conversion in both the global merge regime and the global split regime, and evaluated the bandwidth costs achieved by their constructions.

Since then, several works have further investigated the conversion problem for LRCs. Kong~\cite{kong2024locally} studied locally repairable convertible codes and established lower bounds and matching constructions under the access cost metric. Ge et al.~\cite{ge2026locally} improved the access cost lower bound and proposed more general access-optimal constructions. Shi et al.~\cite{shi2026bounds} further studied
generalized merge-convertible codes involving LRCs and developed new bounds and optimal constructions. Recently, Chopra et al.~\cite{chopra2026bandwidth} derived lower bounds on the bandwidth cost of LRC conversion in the global merge regime and showed that the global merge constructions of Maturana and Rashmi~\cite{maturana2023locally} are
bandwidth-optimal over a broad parameter range.

Prior to this work, however, the global split regime was much less understood. No bandwidth lower bound was known for LRC conversion in this regime, and constructions for which the bandwidth cost had been computed were essentially limited to those of Maturana and Rashmi~\cite{maturana2023locally}. Chopra et al.~\cite{chopra2026bandwidth} explicitly posed the problem of deriving bandwidth lower bounds for LRC conversion in the global split regime and determining whether the constructions in~\cite{maturana2023locally} achieve the optimal bandwidth cost.

\subsection{Main Contributions}

In this paper, we focus on stable optimal-distance LRCCs in the global split regime with \(g^I,g^F\le r\). Our main contributions are summarized as follows.

\begin{itemize}
    \item On the converse side, by extending the information-theoretic approaches of~\cite{chopra2026bandwidth,singhvi2025tight} to this setting, we derive two fundamental constraints on the data downloaded during conversion and use them to obtain lower bounds on the read bandwidth of stable optimal-distance LRCCs. These results hold without any linearity assumption on the initial and final LRCs or on the conversion procedure. To the best of our knowledge, this is the first work to derive read-bandwidth lower bounds for stable optimal-distance LRCCs in the global split regime.

    \item On the achievability side, we provide constructions whose read bandwidth attains the lower bounds established in the converse part throughout the parameter range \(g^I,g^F\le r\). Depending on the relative sizes of \(g^I\) and \(g^F\), we divide the analysis into three cases \(g^F=g^I\), \(g^F>g^I\), and 
    \(g^F<g^I\). For each case, we design the underlying MDS array codes using piggybacking techniques so that they have the repair or alignment properties required by the corresponding conversion procedure. We then show that, over any sufficiently large finite field, these MDS array blocks can be combined into stable optimal-distance LRCCs whose read bandwidth exactly matches the corresponding lower bound.

\end{itemize}
    Combining the converse and achievability results, we characterize the optimal read bandwidth for stable optimal-distance LRCCs in the global split regime over the entire parameter range \(g^I,g^F\le r\). This answers the open problem raised in~\cite{chopra2026bandwidth} for this parameter range. Moreover, by comparing our optimal bandwidth with the bandwidth achieved by the global-split constructions of Maturana and Rashmi~\cite{maturana2023locally}, we show that their constructions are not bandwidth-optimal in general in this regime.

\subsection{Organization of the Paper}

The remainder of this paper is organized as follows. Section~\ref{sec2} reviews the necessary preliminaries on locally repairable codes, MDS array codes, and convertible codes. Section~\ref{sec3} establishes several information-theoretic properties of optimal-distance LRCs, which serve as key tools for the bandwidth lower bound analysis. Section~\ref{sec4} derives read-bandwidth lower bounds for LRC conversion in the global split regime via an information-theoretic approach. Section~\ref{sec5} develops constructions based on MDS array codes with prescribed repair or alignment properties. In particular, we treat the three cases $g^F=g^I$, $g^F>g^I$, and $g^F<g^I$ separately, and provide conversion schemes that attain the lower bounds derived in Section~\ref{sec4}. Section~\ref{sec6} concludes the paper and discusses possible directions for future research.

\section{background and notation}
\label{sec2}
In this section, we introduce the notation used throughout the paper and
review some preliminaries needed for our results.

We use the following notation, largely following the conventions in~\cite{chopra2026bandwidth,singhvi2025tight}. For positive integers $a$ and $i$, define $[a]^i \coloneqq \{a(i-1)+1,a(i-1)+2,\ldots,ai\}$, with $[a]\coloneqq [a]^1$. Let $\mathbb F_q$ denote the finite field of size $q$, where $q$ is a prime power. For an indexed collection $Z=\{Z_i\}_{i\in \mathcal I}$ and a subset $\mathcal A\subseteq \mathcal I$, we write $Z_\mathcal A \coloneqq \{Z_i:i\in \mathcal A\}$. Given a vector $x=(x_i)_{i\in \mathcal I}$, we use $x_\mathcal A$ to denote its projection onto the coordinates indexed by $\mathcal A$. Analogously, given a family of deterministic functions \(\{f_i\}_{i\in\mathcal I}\), we write \(f_{\mathcal A}(Z_{\mathcal A})\coloneqq \{f_i(Z_i):i\in\mathcal A\}\). For random variables $X$ and $Y$, we denote the entropy of $X$ by $H(X)$, the conditional entropy of $X$ given $Y$ by $H(X\mid Y)$, and the mutual information between $X$ and $Y$ by $I(X;Y)$. All entropies are measured in units of $\mathbb F_q$-symbols.

An $[n,k,\alpha]_q$ array code $\mathcal C$ is an injective map \[ \mathcal C:\mathbb F_q^{\alpha k}\to \mathbb F_q^{\alpha n}. \] Unless explicitly stated otherwise, array codes in this paper are not assumed to be linear. We also use $\mathcal C$ to denote the image of this map when no confusion can arise. We identify $\mathbb F_q^{\alpha k}$ with $(\mathbb F_q^\alpha)^k$ and write a message vector as $\bm{m}=(m_1,\ldots,m_k)$, where each $m_i\in\mathbb F_q^\alpha$ is a message block. Given a message vector $\bm{m}$, let $ \bm{c}=\mathcal C(\bm{m})$ be the corresponding codeword and write
\[
     \bm{c}=(c_1,\ldots,c_n), \qquad c_i=(c_{i,1},\ldots,c_{i,\alpha})\in\mathbb F_q^\alpha, \quad i\in[n].
\] Each block $c_i$ is called a code symbol, and its $\alpha$ entries are called subsymbols. For later use, for a node $Z=(z_1,\ldots,z_\alpha)\in\mathbb F_q^\alpha$, we also write $Z[s]\coloneqq z_s$ for its $s$-th subsymbol, where $s\in[\alpha]$. In particular, $c_i[s]=c_{i,s}$. The parameter $\alpha$ is referred to as the sub-packetization level. When $\alpha=1$, the array code model reduces to the usual scalar code model.

For a subset \(\mathcal{S}\subseteq[n]\), let \(\bm{c}_\mathcal{S}\) denote the projection of the codeword \(\bm{c}\) onto the symbols indexed by \(\mathcal{S}\). The punctured code on $\mathcal S$ is denoted by \(\mathcal{C}_\mathcal{S}\coloneqq\{\bm{c}_\mathcal{S}:\bm{c}\in \mathcal{C}\}\). The Hamming distance between array codewords is measured at the symbol level. The minimum distance of an array code is defined with respect to this Hamming distance.

In the storage interpretation, each symbol $c_i$ is stored on a distinct storage node; hence, we use the terms code symbol and storage node interchangeably when discussing repair and conversion procedures. The code $\mathcal C$ is systematic if, possibly after a relabeling of nodes, the first $k$ nodes store the message blocks uncoded, i.e.,
\[ c_i=m_i,\qquad i\in[k]. \]
These $k$ nodes are called information nodes, and the remaining $n-k$ nodes are called parity nodes.

\subsection{Locally Repairable Codes}
We recall the standard definition of $(r,\delta)$-locality for codes that are not necessarily linear.

\begin{definition}[$(r,\delta)$-locality,\cite{kamath2014codes}]
    Let $\mathcal C$ be an $[n,k,\alpha]_q$ array code, and let 
    $r\in [k]$ and $\delta\ge 2$. A subset $\mathcal{S}\subseteq[n]$ is called an $(r,\delta)$-locality set of $\mathcal C$ if
\begin{enumerate}
    \item $|\mathcal{S}|\le r+\delta-1$;
    \item the minimum distance of the punctured code $\mathcal C_\mathcal{S}$ satisfies $d(\mathcal C_\mathcal{S})\ge \delta$.
\end{enumerate}
    The code $\mathcal C$ has $(r,\delta)$ all-symbol locality if every node $t\in[n]$ belongs to at least one $(r,\delta)$-locality set. Furthermore, if $\mathcal C$ is systematic and, without loss of generality, the information nodes are indexed by $[k]$, then $\mathcal C$ has $(r,\delta)$ information locality if every information node $t\in[k]$ belongs to at least one $(r,\delta)$-locality set.
\end{definition}

Next, we introduce the class of locally repairable array codes considered in this paper.

\begin{definition}[$(k,g,r,\ell,\alpha)$-LRC,\cite{maturana2023locally}]
    Suppose that $r\mid k$, set $ \mu=\frac{k}{r}$ and let $  n=k+\mu\ell+g.$ A systematic $[n,k,\alpha]_q$ array code $\mathcal C$ is called a $(k,g,r,\ell,\alpha)$-LRC if it satisfies the following properties:
\begin{enumerate}
    \item $\mathcal C$ has $(r,\ell+1)$ information locality;
    \item The $k$ information nodes are partitioned into $\mu$ disjoint sets $\mathcal I_1,\ldots,\mathcal I_\mu$, each of size $r$. For each $s\in[\mu]$, there is an associated set $\mathcal L_s$ of $\ell$ local parity nodes. These local parity nodes are deterministic functions only of the information nodes in $\mathcal I_s$, and the sets $\mathcal I_s\cup\mathcal L_s$ are pairwise disjoint. Moreover, $\mathcal I_s\cup\mathcal L_s$ is an $(r,\ell+1)$-locality set of $\mathcal C$. We refer to $\mathcal I_s\cup\mathcal L_s$ as the $s$-th local group;
    \item In addition to the information nodes and local parity nodes, each codeword contains $g$ global parity nodes. Each global parity node is a deterministic function of the information nodes.
\end{enumerate}
    When $\alpha$ is clear from the context, we simply call $\mathcal C$ a $(k,g,r,\ell)$-LRC.
\end{definition}

Let $d$ denote the minimum distance of $\mathcal C$. The Singleton-type bound for codes with $(r,\ell+1)$ information locality~\cite{kamath2014codes} gives \[ d \le n-k-\left(\left\lceil\frac{k}{r}\right\rceil-1\right)\ell+1.\] Since $r\mid k$ and $n=k+\mu\ell+g$, this bound reduces to \[d\le g+\ell+1.\]
A $(k,g,r,\ell)$-LRC is called optimal-distance if it meets this bound with equality.

A standard way to construct optimal-distance LRCs is the basic pyramid code construction~\cite{huang2013pyramid}. Starting from an MDS code with $k$ information nodes and $g+\ell$ parity nodes, one keeps $g$ parity nodes as global parity nodes and splits each of the remaining $\ell$ parity nodes into $\mu=\frac{k}{r}$ local parity nodes, one for each local group, while preserving the optimal-distance property.

\subsection{MDS Array Codes and Repair Bandwidth}

We next recall MDS array codes and the repair-bandwidth bounds used in our constructions.

\begin{definition}[MDS array code]
    Let $\mathcal C$ be an $[n,k,\alpha]_q$ array code. We say that $\mathcal C$ is an MDS array code if any $k$ out of the $n$ nodes determine the original message. Equivalently, for every subset $\mathcal{S}\subseteq[n]$ with $|\mathcal{S}|=k$, the projection map $\bm c\mapsto \bm c_{\mathcal S}$ is injective on $\mathcal C$. With respect to the Hamming distance, an MDS array code satisfies\[ d(\mathcal C)=n-k+1.\]
\end{definition}

Let $\mathcal E,\mathcal R\subseteq[n]$ be two disjoint subsets, where $\mathcal E$ denotes the erased nodes and $\mathcal R$ denotes the helper nodes. We denote by $N(\mathcal C,\mathcal E,\mathcal R)$ the minimum number of $\mathbb F_q$-symbols that must be downloaded from the helper nodes in $\mathcal R$ in order to recover the erased nodes in $\mathcal E$. The downloaded symbols may be arbitrary functions of the data stored in the helper nodes.

\begin{theorem}[Centralized repair cut-set bound~\cite{cadambe2013asymptotic}]
\label{thm:cut_set_bound}
    Let $\mathcal C$ be an $[n,k,\alpha]_q$ MDS array code. Let $\mathcal E,\mathcal R\subseteq[n]$ be two disjoint subsets with \[ |\mathcal E|=h,\qquad |\mathcal R|=d.\]
    If $1\le h\le n-k$ and $k\le d\le n-h$, then \[ N(\mathcal C,\mathcal E,\mathcal R) \ge \frac{hd}{h+d-k}\alpha. \]
    If equality holds, we say that the nodes in $\mathcal E$ can be optimally repaired from the helper nodes in $\mathcal R$.
\end{theorem}

Theorem \ref{thm:cut_set_bound} defines optimal repair for a prescribed erased set $\mathcal E$ and a prescribed helper set $\mathcal R$. The following stronger property requires optimal repair for every choice of erased and helper sets of the same sizes.

\begin{definition}[$(h,d)$-optimal repair property]
    Let $\mathcal C$ be an $[n,k,\alpha]_q$ MDS array code. For integers $h$ and $d$ satisfying $ 1\le h\le n-k, k\le d\le n-h, $
    we say that $\mathcal C$ has the $(h,d)$-optimal repair property if, for every erased set $\mathcal E\subseteq[n]$ with $|\mathcal E|=h$ and every helper set $\mathcal R\subseteq[n]\setminus\mathcal E$ with $|\mathcal R|=d$, the nodes in $\mathcal E$ can be optimally repaired from the helper nodes in $\mathcal R$, i.e., \[  N(\mathcal C,\mathcal E,\mathcal R) = \frac{hd}{h+d-k}\alpha .\]
\end{definition}

\begin{theorem}[Existence of MDS array codes with optimal repair
{\cite{ye2017explicit1}}] \label{thm:existence}
Let \(1\le k<n\), set \(\rho=n-k\), and define $s=\operatorname{lcm}(1,2,\ldots,\rho).$ For every finite field \(\mathbb F_q\) with \(q\ge sn\), there exists an \([n,k,s^n]_q\) MDS array code having the \((h,d)\)-optimal repair property simultaneously for all
\[
    1\le h\le \rho,\qquad k\le d\le n-h.
\]
\end{theorem}

\subsection{Convertible Codes}

We next recall the formal model for code conversion.

\begin{definition}[Convertible code]
    A convertible code over $\mathbb F_q$ with parameters $(n^I,k^I;n^F,k^F,\alpha)$ consists of the following objects.
\begin{enumerate}
    \item An initial code $\mathcal C^I:\mathbb F_q^{\alpha k^I}\to\mathbb F_q^{\alpha n^I}$ and a final code $\mathcal C^F:\mathbb F_q^{\alpha k^F}\to\mathbb F_q^{\alpha n^F}$;
    \item Two partitions $\mathcal P^I$ and $\mathcal P^F$ of the common message index set $[M]$, where $M=\operatorname{lcm}(k^I,k^F)$. Each block in $\mathcal P^I$ has size $k^I$, and each block in $\mathcal P^F$ has size $k^F$;
    \item A conversion procedure that, for every message $\bm{m}\in \mathbb F_q^{\alpha M}$, takes the initial codewords $\{\mathcal C^I(\bm{m}_P):P\in\mathcal P^I\}$ as input and produces the final codewords $\{\mathcal C^F(\bm{m}_P):P\in\mathcal P^F\}$ as output, where $\bm{m}_P$ denotes the message blocks indexed by $P$, listed in increasing order of their indices in $[M]$.
\end{enumerate}
\end{definition}

During conversion, the nodes involved are classified into three categories. An \emph{unchanged node} is an initial node whose content is kept unchanged and appears as a node in the final codewords. A \emph{retired node} is an initial node that does not appear in the final codewords after conversion. A \emph{new node} is a node of the final codewords that is generated during conversion and was not present among the initial nodes.

The conversion is carried out by a \emph{conversion coordinator}, which downloads data from the initial nodes and computes the new nodes. A convertible code is called \emph{stable} if the number of unchanged nodes is maximized among all conversion procedures with the same initial and final code parameters. The \emph{read bandwidth cost}, or simply the read bandwidth, denoted by \(\gamma_R\), is the total number of \(\mathbb F_q\)-symbols downloaded by the conversion coordinator during the conversion procedure. Since each node stores \(\alpha\) symbols over \(\mathbb F_q\), we also use the normalized read bandwidth \(\widetilde{\gamma}_R\coloneqq \gamma_R/\alpha\), which measures the read bandwidth in units of full-node contents and allows us to compare constructions with different subpacketization levels. For stable conversions, the number of new nodes is fixed, and hence the write bandwidth is fixed as well. Therefore, under the stability assumption, it suffices to optimize the read bandwidth.

We now specialize the general model of code conversion to locally repairable codes.

\begin{definition}[LRCC] 
    Suppose the initial code $\mathcal C^I$ is a $(k^I,g^I,r^I,\ell^I,\alpha)$-LRC and the final code $\mathcal C^F$ is a $(k^F,g^F,r^F,\ell^F,\alpha)$-LRC. Then the corresponding convertible code is called a locally repairable convertible code, or LRCC, with parameters $(k^I,g^I,r^I,\ell^I;k^F,g^F,r^F,\ell^F,\alpha).$ If both $\mathcal C^I$ and $\mathcal C^F$ are optimal-distance LRCs, then the LRCC is called an optimal-distance LRCC. 
\end{definition}

For an LRCC, let $\mu^I=\frac{k^I}{r^I},\mu^F=\frac{k^F}{r^F}$ denote the numbers of local groups in the initial and final codewords, respectively.

\begin{definition}[Global split regime,\cite{maturana2023locally}] 
    An LRCC is said to be in the global split regime if one initial codeword is converted into $\lambda^F\ge 2$ final codewords and the locality parameters are preserved. More precisely, we have $k^I=\lambda^F k^F, r^I=r^F=r$ and $\ell^I=\ell^F=\ell$, which implies $\mu^I=\lambda^F\mu^F$. 
\end{definition}

For LRCCs in the global split regime, we use the shorthand $(k^I,g^I,r,\ell;k^F,g^F,r,\ell,\alpha)\text{-LRCC}$ and take \(\mathcal P^I=\{[k^I]\}\) and \(\mathcal P^F=\{[k^F]^t:t\in[\lambda^F]\}\) after a possible relabeling of the message indices. 

\begin{remark} \label{rem1}
    In the global split regime, the initial codeword and the collection of final codewords encode the same underlying data. Since the codes are systematic, we identify their information nodes with the same message blocks. Hence, the information nodes of the initial codeword remain unchanged and directly serve as the corresponding information nodes in the final codewords.
\end{remark}

For later comparison, we recall the global-split construction of Maturana and Rashmi. Using the piggybacking framework, they constructed optimal-distance LRCCs in which all information nodes and local parity nodes are preserved during conversion.  In particular, in the global split regime, their construction gives an optimal-distance
$(k^I,g^I,r,\ell;k^F,g^F,r,\ell,\alpha)$-LRCC whose read bandwidth is
\[
\gamma_R^{\mathrm{MR}}=
\begin{cases}
\displaystyle
\lambda^F g^F \frac{(\lambda^F-1)(k^F+\mu^F\ell)+g^I} {(\lambda^F-1)g^F+g^I+\ell(\lambda^F-1)} \alpha, & \text{if } g^F\le g^I,\\[2ex]
\displaystyle
\lambda^F g^F \frac{(k^F+\mu^F\ell)(\lambda^F g^F-g^I)+g^I g^F} {\lambda^F g^F(g^F+\ell)-g^I\ell} \alpha, & \text{otherwise}.
\end{cases}
\] Recent work on the global merge regime raised the corresponding converse problem for the global split regime, namely, to derive lower bounds on the bandwidth of LRCCs in this regime and to determine whether the global-split construction of Maturana and Rashmi is bandwidth-optimal.

\section{Information-Theoretic Characterization of LRC Storage}
\label{sec3}

In this section, we isolate several information-theoretic properties of a single optimal-distance $(k,g,r,\ell,\alpha)$-LRC. These properties are independent of any conversion procedure and will be used in Section~\ref{sec4} to derive the read-bandwidth lower bounds for stable LRCCs in the global split regime.

Consider a distributed storage system encoded using an optimal-distance $(k,g,r,\ell,\alpha)$-LRC, and let $\mu=\frac{k}{r}$. Each codeword is stored across $n=k+\mu\ell+g$ nodes, consisting of $k$ information nodes, $\mu\ell$ local parity nodes, and $g$ global parity nodes. We denote the random variables stored in the information nodes by \(X_j\), \(j\in[k]\), those stored in the local parity nodes by \(L_a\), \(a\in[\mu\ell]\), and those stored in the global parity nodes by \(G_i\), \(i\in[g]\). The local groups are indexed in the natural order: for each \(\tau\in[\mu]\), the \(\tau\)-th local group consists of the information nodes \(X_{[r]^\tau}\) and the local parity nodes \(L_{[\ell]^\tau}\).

We assume that the message blocks stored in the information nodes are independent and uniformly distributed. Thus, for every \(\mathcal S\subseteq[k]\),
\[
    H(X_{\mathcal S})=|\mathcal S|\alpha .
\]

Since each local parity node is a deterministic function only of the information nodes in its own local group, for every \(\tau\in[\mu]\),
\[
    H\bigl(L_{[\ell]^\tau}\mid X_{[r]^\tau}\bigr)=0.
\]

Similarly, since each global parity node is a deterministic function of the information nodes,
\[
    H\bigl(G_{[g]}\mid X_{[k]}\bigr)=0.
\]

The local code induced by each local group has minimum distance \(\ell+1\). Hence, for any \(\mathcal A\subseteq[\ell]^\tau\) and \(\mathcal D\subseteq[r]^\tau\) such that \(|\mathcal A|+|\mathcal D|=r\), the information nodes in the \(\tau\)-th local group can be recovered from \(L_{\mathcal A}\) and \(X_{\mathcal D}\), i.e.,
\[
    H\bigl(X_{[r]^\tau}\mid L_{\mathcal A},X_{\mathcal D}\bigr)=0.
\]

Since the LRC has optimal distance \(g+\ell+1\), any erasure pattern of size at most \(g+\ell\) is recoverable. Therefore, for arbitrary subsets \(\mathcal Q\subseteq[g]\), \(\mathcal R\subseteq[\mu\ell]\), and \(\mathcal S\subseteq[k]\) satisfying $|\mathcal Q|+|\mathcal R|+|\mathcal S|\le g+\ell,$
we have
\[
    H\bigl(G_{\mathcal Q},L_{\mathcal R},X_{\mathcal S} \mid  G_{[g]\setminus\mathcal Q},  L_{[\mu\ell]\setminus\mathcal R},  X_{[k]\setminus\mathcal S}\bigr)=0.
\]

Finally, since every node stores \(\alpha\) symbols over \(\mathbb F_q\),
\[
    H(L_a)\le \alpha,\qquad a\in[\mu\ell], \qquad H(G_i)\le \alpha,\qquad i\in[g].
\]

We next recall several known propositions that will be used in our proofs. The first is the total-entropy identity, which states that the parity nodes introduce no entropy beyond that of the \(k\) independent information nodes.

\begin{proposition}[Total-entropy identity~\cite{chopra2026bandwidth}]
\label{prop:lrc-total-entropy}
    For any optimal-distance \((k,g,r,\ell,\alpha)\)-LRC, the total entropy of the data stored across the \(k\) information nodes, the \(\mu\ell\) local parity nodes, and the \(g\) global parity nodes is
    \[
        H\!\left(X_{[k]},L_{[\mu\ell]},G_{[g]}\right)=k\alpha .
    \]
\end{proposition}

The next proposition isolates an MDS-type structure induced by an optimal-distance LRC. After conditioning on all information nodes outside a fixed local group, the nodes in that local group together with the global parity nodes form an \((r+\ell+g,r)\)-MDS array code.

\begin{proposition}[\cite{chopra2026bandwidth}]
\label{prop:lrc-conditional-mds}
    Fix \(\tau\in[\mu]\). For any \(\mathcal{A}\subseteq[g]\), \(\mathcal{B}\subseteq[\ell]^\tau\), and \(\mathcal{D}\subseteq[r]^\tau\) satisfying \[ |\mathcal{A}|+|\mathcal{B}|+|\mathcal{D}|\le r,\]
    we have \[ H\!\left( G_{\mathcal{A}},L_{\mathcal{B}},X_{\mathcal{D}} \mid X_{[k]\setminus [r]^\tau} \right) = \bigl(|\mathcal{A}|+|\mathcal{B}|+|\mathcal{D}|\bigr)\alpha .\]
\end{proposition}

As an immediate consequence of Proposition~\ref{prop:lrc-conditional-mds}, we obtain the following corollary.

\begin{corollary}[]
\label{cor:lrc-global-uniform-independence}
    Let $\mathcal Q\subseteq[g]$ satisfy $|\mathcal Q|\le r$. Then, for every  $\tau\in[\mu]$, \[ H(G_{\mathcal Q})=|\mathcal Q|\alpha,
\qquad
I\left(G_{\mathcal Q};X_{[k]\setminus [r]^\tau}\right)=0 . \]
\end{corollary}

\begin{proof}
    For any fixed \(\tau\in[\mu]\), Proposition~\ref{prop:lrc-conditional-mds} with \(\mathcal B=\mathcal D=\emptyset\) gives $H\!\left(G_{\mathcal Q}\mid X_{[k]\setminus [r]^\tau}\right)=|\mathcal Q|\alpha .$
    On the other hand, since each global parity node stores at most \(\alpha\) symbols, $H(G_{\mathcal Q})\le \sum_{i\in\mathcal Q}H(G_i) \le |\mathcal Q|\alpha .$
    Since conditioning cannot increase entropy, we must have $H(G_{\mathcal Q})= H\!\left(G_{\mathcal Q}\mid X_{[k]\setminus [r]^\tau}\right) = |\mathcal Q|\alpha .$
    Thus $I\!\left(G_{\mathcal Q};X_{[k]\setminus [r]^\tau}\right)=0.$
\end{proof}

Next, we recall another structural property of optimal-distance LRCs:
each global parity node is a function of all \(k\) information nodes.

\begin{proposition}[\cite{chopra2026bandwidth}]
\label{prop:lrc-global-dependence}
    For any optimal-distance \((k,g,r,\ell,\alpha)\)-LRC, each of the \(g\)
    global parity nodes depends nontrivially on every information node.
\end{proposition}

We now record a general entropy inequality on the conditional mutual information between functions of random variables.

\begin{lemma}[\cite{chopra2026bandwidth}]
\label{lem:conditional-mi-bound}
    Let \(Z_1,Z_2,\ldots,Z_n\) be a collection of random variables, and let \(\{f_i\}_{i\in[n]}\) be a collection of deterministic functions. For any pairwise disjoint subsets \(\mathcal{A},\mathcal{B},\mathcal{D}\subseteq[n]\) and any subsets \(\mathcal{A}'\subseteq\mathcal{A}\), \(\mathcal{B}'\subseteq\mathcal{B}\), if the random variables $f_i(Z_i), i\in(\mathcal{A}\cup\mathcal{B}) \setminus(\mathcal{A}'\cup\mathcal{B}')$
    are conditionally independent given \(f_{\mathcal{D}}(Z_{\mathcal{D}})\),
    then \[ I\!\left( f_{\mathcal{A}}(Z_{\mathcal{A}}); f_{\mathcal{B}}(Z_{\mathcal{B}}) \middle| f_{\mathcal{D}}(Z_{\mathcal{D}}) \right) \leq H\!\left( f_{\mathcal{A}'}(Z_{\mathcal{A}'}) \middle| f_{\mathcal{D}}(Z_{\mathcal{D}}) \right) + H\!\left( f_{\mathcal{B}'}(Z_{\mathcal{B}'}) \middle| f_{\mathcal{D}}(Z_{\mathcal{D}}) \right).\]
\end{lemma}

We next combine Proposition~\ref{prop:lrc-conditional-mds} with Lemma~\ref{lem:conditional-mi-bound} to obtain averaged mutual-information bounds for deterministic functions of the data stored in the nodes. In Section~\ref{sec4}, these functions will be instantiated as the data downloaded during conversion.

\begin{lemma}
\label{lem:lrc-averaged-function-bounds}
    Let $\mathcal C$ be an optimal-distance $(k,g,r,\ell,\alpha)$-LRC with $g\le r$. For each $j\in[k]$, let $f_j$ be a deterministic function of $X_j$; for each $a\in[\mu\ell]$, let $\varphi_a$ be a deterministic function of $L_a$; and for each $i\in[g]$, let $\psi_i$ be a deterministic function of $G_i$.

    For any fixed local group \(\tau\in[\mu]\), 
    \[ \begin{aligned} & I\!\left( \psi_{[g]}(G_{[g]});
    f_{[r]^\tau}(X_{[r]^\tau}), \varphi_{[\ell]^\tau}(L_{[\ell]^\tau}) \middle| X_{[k]\setminus[r]^\tau} \right) \\
    & \le \frac{g+\ell}{r+\ell} \left( \sum_{j\in[r]^\tau} H(f_j(X_j)) + \sum_{a\in[\ell]^\tau} H(\varphi_a(L_a)) \right).
    \end{aligned} \]
    Moreover,
    \[ \begin{aligned} & I\!\left( \psi_{[g]}(G_{[g]}); f_{[k]}(X_{[k]}), \varphi_{[\mu\ell]}(L_{[\mu\ell]}) \right) \\
    & \le \frac{g+\ell}{k+\mu\ell} \left( \sum_{j\in[k]} H(f_j(X_j)) + \sum_{a\in[\mu\ell]} H(\varphi_a(L_a)) \right).
    \end{aligned}\]
\end{lemma}

\begin{proof}
    Fix a local group $\tau\in[\mu]$, and choose subsets $\mathcal D\subseteq[r]^\tau, \mathcal A\subseteq[\ell]^\tau $ such that $|\mathcal D|+|\mathcal A|=g+\ell$. Since $g+\bigl|[r]^\tau\setminus\mathcal D\bigr| +\bigl|[\ell]^\tau\setminus\mathcal A\bigr|=r,$ Proposition~\ref{prop:lrc-conditional-mds}, together with the node-capacity bound, implies that the nodes in $G_{[g]}, X_{[r]^\tau\setminus\mathcal D}, L_{[\ell]^\tau\setminus\mathcal A}$ are conditionally independent given \(X_{[k]\setminus[r]^\tau}\). Hence, their deterministic functions are also conditionally independent given $X_{[k]\setminus[r]^\tau}$. Applying
    Lemma~\ref{lem:conditional-mi-bound}, we obtain \[ \begin{aligned} I\!\left( \psi_{[g]}(G_{[g]}); f_{[r]^\tau}(X_{[r]^\tau}), \varphi_{[\ell]^\tau}(L_{[\ell]^\tau}) \middle| X_{[k]\setminus[r]^\tau} \right) 
    &\le H\!\left( f_{\mathcal D}(X_{\mathcal D}), \varphi_{\mathcal A}(L_{\mathcal A}) \middle| X_{[k]\setminus[r]^\tau} \right)\\
    &\le \sum_{j\in\mathcal D}H(f_j(X_j)) + \sum_{a\in\mathcal A}H(\varphi_a(L_a)).
    \end{aligned}\]
    Averaging this inequality over all pairs of $\mathcal D,\mathcal A$ with $|\mathcal D|+|\mathcal A|=g+\ell$, each of the \(r+\ell\) nodes in the \(\tau\)-th local group appears in \(\binom{r+\ell-1}{g+\ell-1}\) pairs, while the total number of pairs is \(\binom{r+\ell}{g+\ell}\). Therefore, the averaging factor is $\frac{\binom{r+\ell-1}{g+\ell-1}} {\binom{r+\ell}{g+\ell}} = \frac{g+\ell}{r+\ell}$. This proves the first inequality.

    For the second inequality, fix a local group \(\tau\in[\mu]\), and choose subsets \(\mathcal D\subseteq[r]^\tau\) and \(\mathcal A\subseteq[\ell]^\tau\) such that $|\mathcal D|+|\mathcal A|=g+\ell$. Let $\overline{\mathcal D}=[k]\setminus\mathcal D,  \overline{\mathcal A}=[\mu\ell]\setminus\mathcal A$. Since the nodes \(X_{\mathcal D}\) and \(L_{\mathcal A}\) form an erasure pattern of size \(g+\ell\), the optimal-distance property implies that the remaining nodes \(X_{\overline{\mathcal D}}\), \(L_{\overline{\mathcal A}}\), and \(G_{[g]}\) determine the whole codeword. Hence, by Proposition~\ref{prop:lrc-total-entropy}, \[ H\!\left( X_{\overline{\mathcal D}}, L_{\overline{\mathcal A}}, G_{[g]} \right) = k\alpha.\]
    We also claim that
    \[ H\!\left( X_{\overline{\mathcal D}}, L_{\overline{\mathcal A}} \right) = (k-g)\alpha.\]
    Indeed, all local groups except the $\tau$-th one contribute $(\mu-1)r\alpha$ entropy, and the $r-g$ remaining local nodes in the $\tau$-th local group contribute $(r-g)\alpha$ entropy by Proposition~\ref{prop:lrc-conditional-mds}. Since $g\le r$, Corollary~\ref{cor:lrc-global-uniform-independence} gives $H(G_{[g]})=g\alpha.$
    Therefore,
    \[
    \begin{aligned} H\!\left( G_{[g]} \middle|X_{\overline{\mathcal D}},L_{\overline{\mathcal A}}\right) &= H\!\left( X_{\overline{\mathcal D}}, L_{\overline{\mathcal A}}, G_{[g]} \right) - H\!\left( X_{\overline{\mathcal D}}, L_{\overline{\mathcal A}} \right)\\
    &= k\alpha-(k-g)\alpha = g\alpha = H(G_{[g]}).
    \end{aligned}
    \]
    Hence
    \[
    I\!\left( G_{[g]}; X_{\overline{\mathcal D}},L_{\overline{\mathcal A}}\right)=0.
    \]
    By the data-processing inequality,
    \[
    I\!\left( \psi_{[g]}(G_{[g]}); f_{\overline{\mathcal D}}(X_{\overline{\mathcal D}}),  \varphi_{\overline{\mathcal A}}(L_{\overline{\mathcal A}})\right)=0.
    \]
    Thus, by the chain rule,
    \[
    \begin{aligned} I\!\left( \psi_{[g]}(G_{[g]}); f_{[k]}(X_{[k]}), \varphi_{[\mu\ell]}(L_{[\mu\ell]}) \right)
    & = I\!\left( \psi_{[g]}(G_{[g]}); f_{\mathcal D}(X_{\mathcal D}), \varphi_{\mathcal A}(L_{\mathcal A}) \middle| f_{\overline{\mathcal D}}(X_{\overline{\mathcal D}}), \varphi_{\overline{\mathcal A}}(L_{\overline{\mathcal A}}) \right)\\
    &\le H\!\left( f_{\mathcal D}(X_{\mathcal D}), \varphi_{\mathcal A}(L_{\mathcal A}) \right)\\
    &\le \sum_{j\in\mathcal D}H(f_j(X_j)) + \sum_{a\in\mathcal A}H(\varphi_a(L_a)).
    \end{aligned}
\]
Finally, average this inequality over all local groups \(\tau\in[\mu]\) and all choices of \(\mathcal D\subseteq[r]^\tau\), \(\mathcal A\subseteq[\ell]^\tau\) with \(|\mathcal D|+|\mathcal A|=g+\ell\). The total number of choices is $\mu\binom{r+\ell}{g+\ell}$, and each information or local parity node appears in $\binom{r+\ell-1}{g+\ell-1}$ choices. Therefore, the averaging factor is $\frac{\binom{r+\ell-1}{g+\ell-1}} {\mu\binom{r+\ell}{g+\ell}} = \frac{g+\ell}{\mu(r+\ell)} = \frac{g+\ell}{k+\mu\ell}$.
This proves the second inequality.
\end{proof}

\section{Lower Bounds on the Read Bandwidth}
\label{sec4}

In this section, we derive lower bounds on the read bandwidth of stable optimal-distance $(k^I=\lambda^F k^F,g^I,r,\ell;$ $\;k^F,g^F,r,\ell)$ $\text{-LRCC}$ in the global split regime. We first identify the node structure of stable conversions in this regime and introduce the random variables used to model the downloaded data. We then establish two entropy constraints on the downloaded data, which lead to the desired read-bandwidth lower bounds.

We first record a structural property of optimal-distance LRCCs in the
global split regime.

\begin{lemma}[]
\label{lem:split-new-retired}
For any optimal-distance \((k^I=\lambda^F k^F,g^I,r,\ell;\;k^F,g^F,r,\ell)\)-LRCC in the global split regime, the following hold:
\begin{enumerate}
    \item all initial global parity nodes are retired nodes;
    \item all final global parity nodes are new nodes;
    \item every unchanged initial local parity node must serve as a local  parity node for the same set of information nodes in the final codewords.
\end{enumerate}
\end{lemma}

\begin{proof}
    We first prove that all initial global parity nodes are retired in the global split regime. Assume to the contrary that some initial global parity node is not retired. Then this node is unchanged and hence appears as a node in one of the final codewords, say the \(t\)-th final codeword for some \(t\in[\lambda^F]\). Since the final codes are systematic and the information nodes have already been identified with the corresponding message blocks by Remark~\ref{rem1}, this unchanged node cannot be a final information node. Therefore, it must serve as either a local parity node or a global parity node of the \(t\)-th final codeword. In either case, it is a deterministic function of only the \(k^F\) information nodes indexed by \([k^F]^t\). On the other hand, by Proposition~\ref{prop:lrc-global-dependence}, every global parity node of the initial codeword depends nontrivially on every one of the \(k^I\) information nodes. Since \(k^I=\lambda^F k^F\) and \(\lambda^F\ge 2\), there exist initial information nodes outside \([k^F]^t\). This contradicts the fact that the unchanged node is a function only of the information nodes indexed by \([k^F]^t\). Therefore, all initial global parity nodes are retired nodes.

    We next show that all final global parity nodes are new. First suppose \(k^F=r\). Then each final codeword has only one local group and, being optimal-distance, is an \([r+\ell+g^F,r,\alpha]\) MDS array code. Hence the distinction between its local and global parity nodes is only a matter of labeling. Fix the \(t\)-th final codeword. By the first part, no inherited parity node can be an initial global parity node. Moreover, an initial local parity node from a different local group is a nonconstant function of information nodes outside \([k^F]^t\), whereas every node in the \(t\)-th final codeword is a function of \(X_{[k^F]^t}\). Thus the only possible inherited parity nodes are among the \(\ell\) initial local parity nodes associated with \([k^F]^t\). We label these inherited parity nodes as local parity nodes and the remaining parity nodes as global parity nodes. Under this labeling, all final global parity nodes are new.

    It remains to consider the case \(k^F>r\). Suppose, for contradiction, that some final global parity node is unchanged, say in the \(t\)-th final codeword. Since the information nodes are preserved as the corresponding information nodes by Remark~\ref{rem1}, this unchanged node must have originated from an initial parity node. If it was an initial local parity node, then it is a function of at most \(r\) information nodes. This contradicts Proposition~\ref{prop:lrc-global-dependence}, because every global parity node of the \(t\)-th final codeword depends nontrivially on all \(k^F\) information nodes, and \(k^F>r\). If it was an initial global parity node, then it is retired by the first part of the proof. Hence no final global parity node can be unchanged when \(k^F>r\).

    Combining the two cases, we conclude that all final global parity nodes are new nodes.

    Finally, consider an unchanged initial local parity node. Since all final global parity nodes are new, this node must appear as a local parity node in some final codeword. Suppose that it is associated with the initial local group \(\tau\). Then it is a deterministic function only of the information nodes \(X_{[r]^\tau}\). If it were used as a local parity node for another set of \(r\) information nodes in the final codewords, then it would also be a deterministic function only of that set. By the local MDS property, a local parity node depends nontrivially on every information node in its own local group. Hence these two sets of information nodes must coincide. Therefore every unchanged initial local parity node must remain a local parity node associated with the same set of information nodes in the final codewords.
\end{proof}

By Remark~\ref{rem1}, the \(k^I\) information nodes are reused as the corresponding information nodes in the final codewords. Lemma~\ref{lem:split-new-retired} further shows that all initial global parity nodes are retired, all final global parity nodes are new, and every unchanged initial local parity node must remain a local parity node associated with the same set of \(r\) information nodes. Hence, in any optimal-distance LRCC conversion in the global split regime, the number of unchanged nodes is at most $k^I+\mu^I\ell$. Moreover, equality holds if and only if every initial local parity node is preserved; equivalently, if and only if every local group is preserved.

For the global split parameters considered in this paper, the above upper bound is attainable. For instance, the construction of Maturana and Rashmi~\cite{maturana2023locally} preserves all information nodes and all local parity nodes. Hence, a conversion is stable if and only if it attains the upper bound \(k^I+\mu^I\ell\) on the number of unchanged nodes, or equivalently, if and only if every local group is preserved. In the sequel, we therefore consider stable optimal-distance LRCCs in this sense: all local groups are preserved, all initial global parity nodes are retired, and all final global parity nodes are new.

Following the notation in~\cite{singhvi2025tight,chopra2026bandwidth}, we use global indexing for all information nodes, local groups, local parity nodes, and global parity nodes. Set \(\mu=\mu^F=k^F/r\). Then \(\mu^I=\lambda^F\mu\). The initial codeword contains information nodes indexed by \([\lambda^F k^F]\), local groups indexed by \([\lambda^F \mu]\), local parity nodes indexed by \([\lambda^F \mu\ell]\), and initial global parity nodes indexed by \([g^I]\).

For each \(t\in[\lambda^F]\), the \(t\)-th final codeword contains information nodes indexed by \([k^F]^t\), local groups indexed by \([\mu]^t\), local parity nodes indexed by \([\mu\ell]^t\), and final global parity nodes indexed by \([g^F]^t\). For a local group \(\tau\in[\lambda^F \mu]\), we denote its information nodes by \([r]^\tau\) and its local parity nodes by \([\ell]^\tau\).

We use \(X_j\), \(j\in[\lambda^F k^F]\), to denote the random variable stored in the \(j\)-th information node, and \(L_a\), \(a\in[\lambda^F \mu\ell]\), to denote the random variable stored in the \(a\)-th local parity node. The random variable stored in the \(i\)-th initial global parity node is denoted by \(G_i^I\), \(i\in[g^I]\), while the random variable stored in the \(i\)-th final global parity node is denoted by \(G_i^F\), \(i\in[\lambda^F g^F]\). For an index set \(\mathcal S\) in the appropriate ambient set, we write \(X_{\mathcal S}\), \(L_{\mathcal S}\), \(G_{\mathcal S}^I\), and \(G_{\mathcal S}^F\) for the corresponding collections of random variables.

We now model the data downloaded by the conversion coordinator. For each information node $j\in[\lambda^F k^F]$, let $V_j$ be the random variable representing the data downloaded from $X_j$. For each initial global parity node $i\in[g^I]$, let $U_i$ be the data downloaded from \(G_i^I\). For each local parity node $a\in[\lambda^F \mu\ell]$, let $W_a$ be the data downloaded from \(L_a\). Since the downloaded data are deterministic functions of the corresponding stored data, we have
\[ \label{eq:download-read-relation}
    H(V_j\mid X_j)=0,\qquad H(U_i\mid G_i^I)=0,\qquad H(W_a\mid L_a)=0. \]
Let \(\beta_j\), \(\sigma_i\), and \(\delta_a\) denote the numbers of \(\mathbb F_q\)-symbols read from the \(j\)-th information node, the \(i\)-th initial global parity node, and the \(a\)-th local parity node, respectively. Then \[ H(V_j)\le \beta_j\le \alpha,\qquad H(U_i)\le \sigma_i\le \alpha,\qquad H(W_a)\le \delta_a\le \alpha. \]

Since all final global parity nodes are new in a stable conversion, the conversion coordinator must compute them from the downloaded data. Hence we have the conversion-coordinator property
\begin{equation}\label{eq:conver_coor_prop}
    H\!\left( G^F_{[\lambda^F g^F]} \mid V_{[\lambda^F k^F]},U_{[g^I]},W_{[\lambda^F \mu\ell]} \right)=0.
\end{equation}

The read bandwidth $\gamma_R$ is the total number of \(\mathbb F_q\)-symbols read during conversion, namely
\[
    \gamma_R:= \sum_{j\in[\lambda^F k^F]}\beta_j + \sum_{i\in[g^I]}\sigma_i+\sum_{a\in[\lambda^F \mu\ell]}\delta_a.
\]

Before deriving the main entropy constraints, we record a simple consequence of the independence of the final codewords. Since different final codewords are functions of disjoint sets of information nodes, the corresponding final global parity nodes and any deterministic functions of their information and local parity nodes are mutually independent across final codewords. This gives the following conditional-entropy decomposition.

\begin{proposition}
\label{prop:final-codeword-decomposition}
    For each information node $j\in[\lambda^F k^F]$, let $f_j$ be a deterministic function of $X_j$. For each local parity node $a\in[\lambda^F \mu\ell]$, let $\varphi_a$ be a deterministic function of $L_a$.

    For $\mathcal J\subseteq[\lambda^F]$, define
    \[
    [k^F]^{\mathcal J}:=\bigcup_{t\in\mathcal J}[k^F]^t, \qquad [\mu\ell]^{\mathcal J}:=\bigcup_{t\in\mathcal J}[\mu\ell]^t.
    \]
    For $\mathcal S\subseteq[\lambda^F]$, define
    \[
    [g^F]^{\mathcal S}:=\bigcup_{t\in\mathcal S}[g^F]^t.
    \]
    Then, for any $\mathcal S\subseteq\mathcal J\subseteq[\lambda^F]$,
    \[
    \begin{aligned}
    &H\!\left(G^F_{[g^F]^{\mathcal S}} \middle| f_{[k^F]^{\mathcal J}}(X_{[k^F]^{\mathcal J}}), \varphi_{[\mu\ell]^{\mathcal J}}(L_{[\mu\ell]^{\mathcal J}}) \right)\\
    &= \sum_{t\in\mathcal S} H\!\left( G^F_{[g^F]^t} \middle| f_{[k^F]^t}(X_{[k^F]^t}), \varphi_{[\mu\ell]^t}(L_{[\mu\ell]^t}) \right).
    \end{aligned}
    \]
\end{proposition}

\begin{proof}
    Since all nodes in the \(t\)-th final codeword are deterministic functions of the information nodes \(X_{[k^F]^t}\), and the sets \([k^F]^t\), \(t\in[\lambda^F]\), are disjoint and consist of independent information nodes, the random vectors
    \[
    \left( G^F_{[g^F]^t}, f_{[k^F]^t}(X_{[k^F]^t}), \varphi_{[\mu\ell]^t}(L_{[\mu\ell]^t}) \right),\qquad t\in[\lambda^F]
    \]
    are mutually independent. Enumerate \(\mathcal S=\{t_1,\ldots,t_m\}\). By the chain rule,
    \[
    \begin{aligned}
    &H\!\left( G^F_{[g^F]^{\mathcal S}} \middle| f_{[k^F]^{\mathcal J}}(X_{[k^F]^{\mathcal J}}), \varphi_{[\mu\ell]^{\mathcal J}}(L_{[\mu\ell]^{\mathcal J}}) \right)\\
    &= \sum_{b=1}^{m} H\!\left( G^F_{[g^F]^{t_b}} \middle| f_{[k^F]^{\mathcal J}}(X_{[k^F]^{\mathcal J}}), \varphi_{[\mu\ell]^{\mathcal J}}(L_{[\mu\ell]^{\mathcal J}}), G^F_{[g^F]^{t_1}},\ldots,G^F_{[g^F]^{t_{b-1}}} \right).
    \end{aligned}
    \]
    For each \(b \in[m]\), since \(\mathcal S\subseteq\mathcal J\), the conditioning includes the random variables \(f_{[k^F]^{t_b}}(X_{[k^F]^{t_b}})\) and \(\varphi_{[\mu\ell]^{t_b}}(L_{[\mu\ell]^{t_b}})\). All other conditioned variables come from final codewords different from the \(t_b\)-th one, and are independent of the \(t_b\)-th final codeword. Hence
    \[
    \begin{aligned}
    &H\!\left( G^F_{[g^F]^{t_b}} \middle| f_{[k^F]^{\mathcal J}}(X_{[k^F]^{\mathcal J}}), \varphi_{[\mu\ell]^{\mathcal J}}(L_{[\mu\ell]^{\mathcal J}}), G^F_{[g^F]^{t_1}},\ldots,G^F_{[g^F]^{t_{b-1}}} \right)\\
    &= H\!\left( G^F_{[g^F]^{t_b}} \middle| f_{[k^F]^{t_b}}(X_{[k^F]^{t_b}}), \varphi_{[\mu\ell]^{t_b}}(L_{[\mu\ell]^{t_b}}) \right).
    \end{aligned}
    \]
    Summing over $b\in[m]$ gives the desired equality.
\end{proof}

We next establish the first entropy constraint on the downloaded data, which will serve as the key ingredient for the first read-bandwidth lower bound.

\begin{lemma} \label{lem: bound1}
    Suppose \(g^F\le r\). Then, for any stable optimal-distance \((k^I=\lambda^F k^F,g^I,r,\ell;\;k^F,g^F,r,\ell)\)-LRCC in the global split regime, it holds that
    \[
    H(U_{[g^I]})+\frac{g^F+\ell}{k^F+\mu\ell}\left(\sum_{j\in[\lambda^F k^F]} H(V_j)+\sum_{a\in[\lambda^F \mu\ell]} H(W_a)\right)\ge\lambda^F g^F\alpha.
    \]
\end{lemma}

\begin{proof}
    By the conversion coordinator property~\eqref{eq:conver_coor_prop}, all final global parity nodes are determined by the downloaded data \(U_{[g^I]},V_{[\lambda^F k^F]},W_{[\lambda^F \mu\ell]}\). Hence
    \[
    \begin{aligned}
    H(U_{[g^I]}) &\ge H\!\left( U_{[g^I]} \mid V_{[\lambda^F k^F]},W_{[\lambda^F \mu\ell]} \right)\\
    &\ge I\!\left( U_{[g^I]}; G^F_{[\lambda^F g^F]} \mid V_{[\lambda^F k^F]},W_{[\lambda^F \mu\ell]} \right)\\
    &= H\!\left( G^F_{[\lambda^F g^F]} \mid V_{[\lambda^F k^F]},W_{[\lambda^F \mu\ell]} \right).
    \end{aligned}
    \]
    By Proposition~\ref{prop:final-codeword-decomposition}, applied with \(\mathcal S=\mathcal J=[\lambda^F]\), \(f_j(X_j)=V_j\), and \(\varphi_a(L_a)=W_a\), we have
    \[
    \begin{aligned}
    H\!\left( G^F_{[\lambda^F g^F]} \mid V_{[\lambda^F k^F]},W_{[\lambda^F \mu\ell]} \right)
    = \sum_{t\in[\lambda^F]} H\!\left( G^F_{[g^F]^t} \mid V_{[k^F]^t},W_{[\mu\ell]^t} \right).
    \end{aligned}
    \]
    For each \(t\in[\lambda^F]\), since \(g^F\le r\), Corollary~\ref{cor:lrc-global-uniform-independence} gives
    \[
    H\!\left(G^F_{[g^F]^t}\right)=g^F\alpha .
    \]
    Therefore,
    \[
    \begin{aligned}
    H\!\left( G^F_{[g^F]^t} \mid V_{[k^F]^t},W_{[\mu\ell]^t} \right)
    = g^F\alpha - I\!\left( G^F_{[g^F]^t}; V_{[k^F]^t},W_{[\mu\ell]^t} \right).
    \end{aligned}
    \]
    Applying Lemma~\ref{lem:lrc-averaged-function-bounds} to the \(t\)-th final LRC, with \(\psi_i(G_i^F)=G_i^F\), \(f_j(X_j)=V_j\), and \(\varphi_a(L_a)=W_a\), yields
    \[
    \begin{aligned}
    I\!\left( G^F_{[g^F]^t}; V_{[k^F]^t},W_{[\mu\ell]^t} \right)
    \le \frac{g^F+\ell}{k^F+\mu\ell} \left( \sum_{j\in[k^F]^t}H(V_j) + \sum_{a\in[\mu\ell]^t}H(W_a) \right).
    \end{aligned}
    \]
    Combining the above inequalities gives
    \[
    \begin{aligned}
    H(U_{[g^I]}) &\ge \lambda^F g^F\alpha
    - \frac{g^F+\ell}{k^F+\mu\ell} \sum_{t\in[\lambda^F]} \left( \sum_{j\in[k^F]^t}H(V_j) + \sum_{a\in[\mu\ell]^t}H(W_a) \right)\\
    &= \lambda^F g^F\alpha - \frac{g^F+\ell}{k^F+\mu\ell} \left( \sum_{j\in[\lambda^F k^F]}H(V_j) + \sum_{a\in[\lambda^F \mu\ell]}H(W_a) \right),
    \end{aligned}
    \]
    where the last equality follows because \(\{[k^F]^t:t\in[\lambda^F]\}\) partitions \([\lambda^F k^F]\) and \(\{[\mu\ell]^t:t\in[\lambda^F]\}\) partitions \([\lambda^F \mu\ell]\). Rearranging gives the desired inequality.
\end{proof}

We now derive the first read-bandwidth lower bound from Lemma~\ref{lem: bound1}.

\begin{theorem} \label{thm:bound1}
    Suppose \(g^F\le r\). For any stable optimal-distance \((k^I=\lambda^F k^F,g^I,r,\ell;\;k^F,g^F,r,\ell)\)-LRCC in the global split regime, the read bandwidth satisfies
    \[
    \gamma_R \ge \lambda^F\frac{g^F}{g^F+\ell}(k^F+\mu\ell)\alpha - g^I\left(\frac{k^F+\mu\ell}{g^F+\ell}-1\right)\alpha .
    \]
\end{theorem}

\begin{proof}
    Let $S :=\sum_{j\in[\lambda^F k^F]}H(V_j) + \sum_{a\in[\lambda^F \mu\ell]}H(W_a), T := H(U_{[g^I]}),$ and set $c:=\frac{g^F+\ell}{k^F+\mu\ell}.$ By the definition of read bandwidth and the entropy bounds on the downloaded data, we have
    \begin{align*}
    \gamma_R &= \sum_{i\in[g^I]}\sigma_i + \sum_{j\in[\lambda^F k^F]}\beta_j + \sum_{a\in[\lambda^F \mu\ell]}\delta_a \\
    &\ge \sum_{i\in[g^I]}H(U_i) + \sum_{j\in[\lambda^F k^F]}H(V_j) + \sum_{a\in[\lambda^F \mu\ell]}H(W_a) \\
    &\ge T+S.
    \end{align*}
    By Lemma~\ref{lem: bound1},
    \[
    T+cS\ge \lambda^F g^F\alpha .
    \]
    Moreover, since \(g^F\le r\) and \(k^F=\mu r\), we have
    \[
    c= \frac{g^F+\ell}{k^F+\mu\ell} = \frac{g^F+\ell}{\mu(r+\ell)} \le 1.
    \]
    Thus \(0<c\le 1\), and hence
    \begin{align*}
    T+S &= \frac{1}{c}(T+cS) - \left(\frac{1}{c}-1\right)T \\
    &\ge \frac{\lambda^F g^F\alpha}{c} - \left(\frac{1}{c}-1\right)T.
    \end{align*}
    For each \(i\in[g^I]\), the downloaded variable \(U_i\) is a deterministic function of \(G_i^I\), and each initial global parity node stores at most \(\alpha\) \(\mathbb F_q\)-symbols. Hence
    \[
    H(U_i)\le H(G_i^I)\le \alpha .
    \]
    By subadditivity,
    \[
    T=H(U_{[g^I]}) \le \sum_{i\in[g^I]}H(U_i) \le g^I\alpha .
    \]
    Combining the above inequalities yields
    \begin{align*}
    \gamma_R &\ge T+S \ge \frac{\lambda^F g^F\alpha}{c} - \left(\frac{1}{c}-1\right)g^I\alpha \\
    &= \lambda^F\frac{g^F}{g^F+\ell}(k^F+\mu\ell)\alpha - g^I\left(\frac{k^F+\mu\ell}{g^F+\ell}-1\right)\alpha.
    \end{align*}
    This completes the proof.
\end{proof}

We next derive a lower bound on
\[
\sum_{j\in[\lambda^F k^F]}H(V_j) + \sum_{a\in[\lambda^F \mu\ell]}H(W_a),
\]
which, together with Lemma~\ref{lem: bound1}, will yield the second
read-bandwidth lower bound.

\begin{lemma} \label{lem:download-vw-lower-bound}
    Suppose \(g^F\le r\) and \(g^I\le r\). Then, for any stable optimal-distance $(k^I=\lambda^F k^F,g^I,r,\ell;\;k^F,g^F,r,\ell) \text{-LRCC}$ in the global split regime, we have
    \[
    \sum_{j\in[\lambda^F k^F]}H(V_j) + \sum_{a\in[\lambda^F \mu\ell]}H(W_a) \ge \frac{ \lambda^F(\lambda^F-1)g^F(k^F+\mu\ell) }{ (\lambda^F-1)(g^F+\ell)+g^I+\ell }\alpha.
    \]
\end{lemma}

\begin{proof}
    For each local group \(\tau\in[\lambda^F \mu]\), define
    \[
    R_\tau:=\big(V_{[r]^\tau},W_{[\ell]^\tau}\big), \qquad \Delta_\tau:= \sum_{j\in[r]^\tau}H(V_j) + \sum_{a\in[\ell]^\tau}H(W_a).
    \]
    For \(s\in[\lambda^F]\) and \(\tau\in[\mu]^s\), write
    \[
    X_{s,\bar\tau}:=X_{[k^F]^s\setminus [r]^\tau}, \qquad X_{\bar\tau}:=X_{[\lambda^F k^F]\setminus [r]^\tau}.
    \]
    We first record two local mutual-information bounds. Applying Lemma~\ref{lem:lrc-averaged-function-bounds} to the \(s\)-th final LRC, with \(\psi_i(G_i^F)=G_i^F\), \(f_j(X_j)=V_j\), and \(\varphi_a(L_a)=W_a\), gives, for every \(\tau\in[\mu]^s\),
    \begin{equation} 
    I\!\left(G^F_{[g^F]^s};R_\tau\mid X_{s,\bar{\tau}}\right) \le \frac{g^F+\ell}{r+\ell}\Delta_\tau. \label{eq:final-local-mi} 
    \end{equation}
    Applying Lemma~\ref{lem:lrc-averaged-function-bounds} to the initial LRC, with \(\psi_i(G_i^I)=U_i\), \(f_j(X_j)=V_j\), and \(\varphi_a(L_a)=W_a\), gives, for every \(\tau\in[\lambda^F\mu]\),
    \begin{equation}
    I\!\left(U_{[g^I]};R_\tau\mid X_{\bar\tau}\right) \le \frac{g^I+\ell}{r+\ell}\Delta_\tau .
    \label{eq:initial-local-mi}
    \end{equation}
    Now fix a tuple
    \[
    \mathcal{T}=(\tau_1,\ldots,\tau_{\lambda^F}),\qquad \tau_w\in[\mu]^w,
    \]
    and define
    \[
    B_{\mathcal T}:= \bigl(X_{w,\bar{\tau}_w},R_{\tau_w}:w\in[\lambda^F]\bigr).
    \]
    We claim that \(B_{\mathcal T}\) determines \(V_{[\lambda^F k^F]}\) and \(W_{[\lambda^F\mu\ell]}\). Indeed, for the selected local group \(\tau_w\), the variables \(V_{[r]^{\tau_w}}\) and \(W_{[\ell]^{\tau_w}}\) are included in \(R_{\tau_w}\). For every other local group in the \(w\)-th final codeword, all its information nodes are contained in \(X_{w,\bar{\tau}_w}\), and hence both the corresponding downloaded information-node data and local-parity-node data are determined.

    Fix \(s\in[\lambda^F]\), and denote by
    \[
    G^F_{\bar s}:=G^F_{\bigcup_{w\in[\lambda^F]\setminus\{s\}}[g^F]^w}
    \]
    the collection of final global parities outside the \(s\)-th final codeword. By the conversion-coordinator property (\ref{eq:conver_coor_prop}), we have
    \[
    H\!\left(G^F_{\bar s}\mid U_{[g^I]},B_\mathcal{T}\right)=0.
    \]
    Therefore,
    \begin{equation}
    H\!\left(G^F_{\bar s}\mid B_\mathcal{T}\right) = I\!\left(U_{[g^I]};G^F_{\bar s}\mid B_\mathcal{T}\right).
    \label{eq:main-identity}
    \end{equation}

    We lower bound the left-hand side of \eqref{eq:main-identity}. Apply Proposition~\ref{prop:final-codeword-decomposition} with \(\mathcal S=[\lambda^F]\setminus\{s\}\) and \(\mathcal J=[\lambda^F]\), where, for each \(w\in[\lambda^F]\), the functions are chosen so that the conditioning variables in the \(w\)-th final codeword are exactly \(X_{w,\bar{\tau}_w}\) and \(R_{\tau_w}\). Then
    \[
    \begin{aligned}
    H\!\left(G^F_{\bar s}\mid B_\mathcal{T}\right) &= \sum_{w\in[\lambda^F]\setminus\{s\}} H\!\left( G^F_{[g^F]^w}\mid X_{w,\bar\tau_w},R_{\tau_w} \right)\\
    &= \sum_{w\in[\lambda^F]\setminus\{s\}} \Big[ H\!\left( G^F_{[g^F]^w}\mid X_{w,\bar\tau_w} \right) - I\!\left( G^F_{[g^F]^w};R_{\tau_w}\mid X_{w,\bar\tau_w} \right) \Big]\\
    &\ge \sum_{w\in[\lambda^F]\setminus\{s\}} \left( g^F\alpha - \frac{g^F+\ell}{r+\ell}\Delta_{\tau_w} \right),
    \end{aligned}
    \]
    where the last step follows from Corollary~\ref{cor:lrc-global-uniform-independence}, which gives \(H(G^F_{[g^F]^w}\mid X_{w,\bar{\tau}_w})=g^F\alpha\), and from \eqref{eq:final-local-mi}.

    We next upper bound the right-hand side of \eqref{eq:main-identity}. Since \(X_{\bar{\tau}_s}\) contains all information nodes outside the local group \(\tau_s\), it determines all final global parity nodes outside the \(s\)-th final codeword and all components of \(B_{\mathcal T}\) except \(R_{\tau_s}\). Hence \((G^F_{\bar{s}},B_{\mathcal T})\) is a deterministic function of \((X_{\bar{\tau}_s},R_{\tau_s})\). Therefore,
    \[
    \begin{aligned}
    I\!\left(U_{[g^I]};G^F_{\bar s}\mid B_\mathcal{T}\right) &\le I\!\left(U_{[g^I]};G^F_{\bar s},B_\mathcal{T}\right)\\
    &\le I\!\left(U_{[g^I]};X_{\bar\tau_s},R_{\tau_s}\right)\\
    &= I\!\left(U_{[g^I]};X_{\bar\tau_s}\right) + I\!\left(U_{[g^I]};R_{\tau_s}\mid X_{\bar\tau_s}\right).
    \end{aligned}
    \]
    By Corollary~\ref{cor:lrc-global-uniform-independence} applied to the initial LRC and by the data-processing inequality,
    \[
    I\!\left(U_{[g^I]};X_{\bar\tau_s}\right)=0.
    \]
    Combining this with \eqref{eq:initial-local-mi}, we obtain
    \[
    I\!\left(U_{[g^I]};G^F_{\bar s}\mid B_T\right) \le \frac{g^I+\ell}{r+\ell}\Delta_{\tau_s}.
    \]

    Combining the preceding lower and upper bounds with \eqref{eq:main-identity}, we obtain, for every \(s\in[\lambda^F]\) and every tuple \(\mathcal T\),
    \[
    \sum_{w\in[\lambda^F]\setminus\{s\}} \left( g^F\alpha - \frac{g^F+\ell}{r+\ell}\Delta_{\tau_w} \right) \le \frac{g^I+\ell}{r+\ell}\Delta_{\tau_s}.
    \]
    Equivalently,
    \begin{equation}
    \sum_{w\in[\lambda^F]\setminus\{s\}} \frac{g^F+\ell}{r+\ell}\Delta_{\tau_w} + \frac{g^I+\ell}{r+\ell}\Delta_{\tau_s} \ge (\lambda^F-1)g^F\alpha .
    \label{eq:tuple-ineq}
    \end{equation}

    Finally, sum \eqref{eq:tuple-ineq} over all \(s\in[\lambda^F]\) and over all tuples \(\mathcal T=(\tau_1,\ldots,\tau_{\lambda^F})\) satisfying \(\tau_w\in[\mu]^w\) for every \(w\in[\lambda^F]\). Fix a final codeword \(w\in[\lambda^F]\) and a local group \(\tau\in[\mu]^w\). The term \(\Delta_\tau\) appears in \(\mu^{\lambda^F-1}\) tuples. For each such tuple, its total coefficient after summing over all choices of \(s\) is $\frac{(\lambda^F-1)(g^F+\ell)+g^I+\ell}{r+\ell}$. Thus,
    \[
    \mu^{\lambda^F-1} \frac{(\lambda^F-1)(g^F+\ell)+g^I+\ell}{r+\ell} \sum_{w\in[\lambda^F]}\sum_{\tau\in[\mu]^w}\Delta_\tau \ge \lambda^F \mu^{\lambda^F}(\lambda^F-1)g^F\alpha .
    \]
    Since
    \[ 
    \sum_{w\in[\lambda^F]}\sum_{\tau\in[\mu]^w}\Delta_\tau = \sum_{j\in[\lambda^F k^F]}H(V_j) + \sum_{a\in[\lambda^F \mu\ell]}H(W_a),
    \]
    we get
    \[
    \sum_{j\in[\lambda^F k^F]}H(V_j) + \sum_{a\in[\lambda^F \mu\ell]}H(W_a) \ge \frac{ \lambda^F(\lambda^F-1)g^F \mu(r+\ell) }{ (\lambda^F-1)(g^F+\ell)+g^I+\ell }\alpha .
    \]
    Using \(\mu(r+\ell)=k^F+\mu\ell\), the desired bound follows.
\end{proof}

Combining Lemma~\ref{lem: bound1} with Lemma~\ref{lem:download-vw-lower-bound}, we obtain the following second read-bandwidth lower bound for stable optimal-distance LRCCs in the global split regime.

\begin{theorem} \label{thm:bound2}
    Suppose \(g^F\le r\) and \(g^I\le r\). For any stable optimal-distance \((k^I=\lambda^F k^F,g^I,r,\ell;\;k^F,g^F,r,\ell)\)-LRCC in the global split regime, the read bandwidth satisfies
    \[
    \gamma_R \ge \lambda^F g^F \frac{ (\lambda^F-1)(k^F+\mu\ell)+g^I+\ell }{ (\lambda^F-1)(g^F+\ell)+g^I+\ell } \alpha.
    \]
\end{theorem}

\begin{proof}
    Let
    \[
    S := \sum_{j\in[\lambda^F k^F]}H(V_j) + \sum_{a\in[\lambda^F \mu\ell]}H(W_a), \qquad T:=H(U_{[g^I]}).
    \]
    By the definition of read bandwidth and the entropy bounds on the downloaded data,
    \begin{align*}
    \gamma_R &= \sum_{i\in[g^I]}\sigma_i + \sum_{j\in[\lambda^F k^F]}\beta_j + \sum_{a\in[\lambda^F \mu\ell]}\delta_a \\
    &\ge \sum_{i\in[g^I]}H(U_i) + \sum_{j\in[\lambda^F k^F]}H(V_j) + \sum_{a\in[\lambda^F \mu\ell]}H(W_a) \\
    &\ge T+S.
    \end{align*}
    Set \(c=(g^F+\ell)/(k^F+\mu\ell)\). Since \(g^F\le r\) and \(k^F=\mu r\), we have \(0<c\le 1\). Hence
    \[
    \gamma_R\ge T+S = T+cS+(1-c)S .
    \]
    By Lemma \ref{lem: bound1},
    \[
    T+cS\ge \lambda^F g^F\alpha .
    \]
    Moreover, by Lemma~\ref{lem:download-vw-lower-bound},
    \[
    S \ge \frac{ \lambda^F(\lambda^F-1)g^F(k^F+\mu\ell) }{(\lambda^F-1)(g^F+\ell)+g^I+\ell}\alpha .
    \]
    Combining the above inequalities gives
    \begin{align*}
    \gamma_R &\ge \lambda^F g^F\alpha + \left( 1-\frac{g^F+\ell}{k^F+\mu\ell} \right) \frac{ \lambda^F(\lambda^F-1)g^F(k^F+\mu\ell) }{ (\lambda^F-1)(g^F+\ell)+g^I+\ell } \alpha \\
    &= \lambda^F g^F\alpha + \frac{ \lambda^F(\lambda^F-1)g^F(k^F+\mu\ell-g^F-\ell) }{ (\lambda^F-1)(g^F+\ell)+g^I+\ell} \alpha \\
    &= \lambda^F g^F \frac{ (\lambda^F-1)(k^F+\mu\ell)+g^I+\ell }{ (\lambda^F-1)(g^F+\ell)+g^I+\ell } \alpha .
    \end{align*}
    This completes the proof.
\end{proof}

Theorems~\ref{thm:bound1} and~\ref{thm:bound2} together imply the following piecewise lower bound on the read bandwidth.

\begin{theorem} 
\label{thm:final_thm}
Suppose \(g^I,g^F\le r\). For any stable optimal-distance \((k^I=\lambda^F k^F,g^I,r,\ell;\;k^F,g^F,r,\ell)\)-LRCC in the global split regime, the read bandwidth satisfies 
\[
\gamma_R \ge \begin{cases} \displaystyle \lambda^F g^F \frac{ (\lambda^F-1)(k^F+\mu\ell)+g^I+\ell }{ (\lambda^F-1)(g^F+\ell)+g^I+\ell } \alpha, & \text{if } g^F\le g^I, \\[2.2ex] \displaystyle \lambda^F\frac{g^F}{g^F+\ell}(k^F+\mu\ell)\alpha - g^I\left(\frac{k^F+\mu\ell}{g^F+\ell}-1\right)\alpha, & \text{if } g^F>g^I. \end{cases} 
\]
\end{theorem} 

\begin{proof} 
By the preceding two theorems, we have 
\[
\gamma_R\ge A \qquad\text{and}\qquad \gamma_R\ge B, 
\] 
where 
\[ 
A:= \lambda^F g^F \frac{ (\lambda^F-1)(k^F+\mu\ell)+g^I+\ell }{ (\lambda^F-1)(g^F+\ell)+g^I+\ell } \alpha 
\]
and 
\[ 
B:= \lambda^F\frac{g^F}{g^F+\ell}(k^F+\mu\ell)\alpha - g^I\left(\frac{k^F+\mu\ell}{g^F+\ell}-1\right)\alpha . 
\]
Thus \(\gamma_R\ge\max\{A,B\}\). It remains to compare \(A\) and \(B\).
A direct calculation gives
\[
B-A = \frac{ (g^F-g^I)(\lambda^F\ell+g^I)(k^F+\mu\ell-g^F-\ell) }{ (g^F+\ell)\bigl((\lambda^F-1)(g^F+\ell)+g^I+\ell\bigr) } \alpha . 
\]
Since \(k^F+\mu\ell=\mu(r+\ell)\) and \(g^F\le r\), we have \(k^F+\mu\ell-g^F-\ell\ge 0\). All other factors except \(g^F-g^I\) are nonnegative, and the denominator is positive. Hence \(B-A\le 0\) when \(g^F\le g^I\), while \(B-A\ge 0\) when \(g^F\ge g^I\). Therefore, when \(g^F\le g^I\), the bound \(A\) is at least as strong as \(B\), and when \(g^F>g^I\), the bound \(B\) is at least as strong as \(A\). This gives the desired piecewise bound. \end{proof}

\begin{remark}
\label{rem:normalized-lower-bound}
    Theorem~\ref{thm:final_thm} is stated for the unnormalized read bandwidth \(\gamma_R\). Dividing both sides by \(\alpha\) gives the corresponding lower bound on the normalized read bandwidth \(\widetilde{\gamma}_R=\gamma_R/\alpha\). In the next section, we construct stable optimal-distance LRCCs whose normalized read bandwidth attains this bound over sufficiently large finite fields.
\end{remark}

\begin{remark}
Let $\gamma_R^{\mathrm{MR}}$ denote the read bandwidth achieved by the global split constructions of Maturana and Rashmi~\cite{maturana2023locally}, and let $\gamma_R^{\mathrm{LB}}$ denote the right-hand side of Theorem~\ref{thm:final_thm}. Set
\[
    S:=k^F+\mu\ell, \qquad D:=(\lambda^F-1)(g^F+\ell)+g^I.
\]
A direct subtraction gives
\[
\gamma_R^{\mathrm{MR}}-\gamma_R^{\mathrm{LB}} =
\begin{cases}
\displaystyle
\frac{ \lambda^F g^F\ell(\lambda^F-1) (S-g^F-\ell) }{ D(D+\ell) }\alpha,
& g^F\le g^I, \\[3mm]
\displaystyle
\frac{ g^I\ell(\lambda^F g^F-g^I) (S-g^F-\ell) }{ (g^F+\ell) \bigl(\lambda^F g^F(g^F+\ell)-g^I\ell\bigr) }\alpha,
& g^F>g^I.
\end{cases}
\]
Since $k^F=\mu r$ and $g^F\le r$,
\[
    S-g^F-\ell  =(\mu-1)(r+\ell)+(r-g^F)\ge0.
\]
Therefore, $\gamma_R^{\mathrm{MR}}\ge\gamma_R^{\mathrm{LB}}$. Under the standing assumptions $\lambda^F\ge2$ and $g^I,\ell>0$, equality holds if and only if $\mu=1$ and $g^F=r$; otherwise, the inequality is strict. Since the next section presents constructions that attain the lower bound in Theorem~\ref{thm:final_thm}, the global-split constructions of Maturana and Rashmi are not read-bandwidth-optimal in general in this parameter regime.
\end{remark}

\section{Optimal Constructions of Locally Repairable Convertible Codes in the Global Split Regime}
\label{sec5}
In this section, we construct stable optimal-distance LRCCs that attain the read-bandwidth lower bounds derived in Section~\ref{sec4}. Throughout this section, we assume \(g^I,g^F\le r\) and keep the notation \(\mu=\frac{k^F}{r}\) and \(\lambda=\lambda^F\). According to the relative sizes of \(g^I\) and \(g^F\), we treat the three cases \(g^F=g^I\), \(g^F>g^I\), and \(g^F<g^I\) separately. In each case, the construction is based on MDS array codes with prescribed repair or alignment properties, and we show that the resulting stable LRCC attains the corresponding branch of the lower bound in Theorem~\ref{thm:final_thm}.

\subsection{\texorpdfstring{$g^F=g^I$}{gF = gI}}
\label{subsec5.1}

\subsubsection{An illustrative example}

We first present an illustrative example for the case \(g^F=g^I\). Consider a global split conversion from an initial \((8,1,2,1,\alpha=2)\)-LRC codeword to two final \((4,1,2,1,\alpha=2)\)-LRC codewords. Here \(\lambda^F=2\), \(\mu=2\), and \(g^I=g^F=1\). We take \(\alpha=2\) and work over \(\mathbb F_3\).

Let the eight information nodes be $X_{2i-1}=(a_i,b_i)^\top, X_{2i}=(c_i,d_i)^\top, i\in[4]$. For the \(i\)-th local group, define the local parity node $L_i=(a_i+2c_i+b_i+d_i,\; b_i+2d_i)^\top$ and introduce the virtual global parity node $P_i=(a_i+c_i,\; b_i+d_i)^\top$. The two final global parity nodes are \(G_1^F=P_1+P_2\) and \(G_2^F=P_3+P_4\), while the initial global parity node is \(G_1^I=G_1^F+G_2^F=P_1+P_2+P_3+P_4\).

The initial codeword and the two final codewords are shown in Table~\ref{tab:example_g_eq_initial}. The yellow entries in the initial codeword are the symbols read during conversion.

\begin{table*}[t]
\centering
\caption{The \((8,1,2,1,\alpha=2)\)-to-\((4,1,2,1,\alpha=2)\) split conversion example. Yellow entries are read during conversion.}
\label{tab:example_g_eq_initial}
\begingroup
\scriptsize
\setlength{\tabcolsep}{2.2pt}
\renewcommand{\arraystretch}{0.52}
\setlength{\fboxsep}{0.5pt}
\textbf{Initial codeword.}\par\vspace{0.2ex}
\begin{adjustbox}{max width=0.86\textwidth}
\begin{tabular}{c|c|c|c}
\hline
Pos. & Node & \(1\) & \(2\)\\
\hline
1  & \(X_1\) & \readcell{\(a_1\)} & \(b_1\)\\
2  & \(X_2\) & \readcell{\(c_1\)} & \(d_1\)\\
3  & \(L_1\) & \readcell{\(a_1+2c_1+b_1+d_1\)} & \(b_1+2d_1\)\\
\hline
4  & \(X_3\) & \readcell{\(a_2\)} & \(b_2\)\\
5  & \(X_4\) & \readcell{\(c_2\)} & \(d_2\)\\
6  & \(L_2\) & \readcell{\(a_2+2c_2+b_2+d_2\)} & \(b_2+2d_2\)\\
\hline
7  & \(X_5\) & \(a_3\) & \(b_3\)\\
8  & \(X_6\) & \(c_3\) & \(d_3\)\\
9  & \(L_3\) & \(a_3+2c_3+b_3+d_3\) & \(b_3+2d_3\)\\
\hline
10 & \(X_7\) & \(a_4\) & \(b_4\)\\
11 & \(X_8\) & \(c_4\) & \(d_4\)\\
12 & \(L_4\) & \(a_4+2c_4+b_4+d_4\) & \(b_4+2d_4\)\\
\hline
13 & \(G_1^I\) &
\readcell{\(a_1+c_1+a_2+c_2+a_3+c_3+a_4+c_4\)} &
\readcell{\(b_1+d_1+b_2+d_2+b_3+d_3+b_4+d_4\)}\\
\hline
\end{tabular}
\end{adjustbox}
\vspace{0.8mm}

\textbf{Final codewords.}\par\vspace{0.2ex}
\begin{adjustbox}{max width=\textwidth}
\begin{tabular}{c|c|c|c||c|c|c|c}
\hline
Pos. & Node in \(C_1^F\) & \(1\) & \(2\)
& Pos. & Node in \(C_2^F\) & \(1\) & \(2\)\\
\hline
1 & \(X_1\) & \(a_1\) & \(b_1\)
  & 1 & \(X_5\) & \(a_3\) & \(b_3\)\\
2 & \(X_2\) & \(c_1\) & \(d_1\)
  & 2 & \(X_6\) & \(c_3\) & \(d_3\)\\
3 & \(L_1\) & \(a_1+2c_1+b_1+d_1\) & \(b_1+2d_1\)
  & 3 & \(L_3\) & \(a_3+2c_3+b_3+d_3\) & \(b_3+2d_3\)\\
\hline
4 & \(X_3\) & \(a_2\) & \(b_2\)
  & 4 & \(X_7\) & \(a_4\) & \(b_4\)\\
5 & \(X_4\) & \(c_2\) & \(d_2\)
  & 5 & \(X_8\) & \(c_4\) & \(d_4\)\\
6 & \(L_2\) & \(a_2+2c_2+b_2+d_2\) & \(b_2+2d_2\)
  & 6 & \(L_4\) & \(a_4+2c_4+b_4+d_4\) & \(b_4+2d_4\)\\
\hline
7 & \(G_1^F\) &
\(a_1+c_1+a_2+c_2\) &
\(b_1+d_1+b_2+d_2\)
  & 7 & \(G_2^F\) &
\(a_3+c_3+a_4+c_4\) &
\(b_3+d_3+b_4+d_4\)\\
\hline
\end{tabular}
\end{adjustbox}
\endgroup
\end{table*}

A Magma computation gives \(d(C^I)=3\) and \(d(C^F)=3\). Hence \(C^I\) and \(C^F\) are both optimal-distance LRCs.

The conversion preserves all information nodes and local parity nodes. The coordinator first reads the whole initial global parity node \(G_1^I\), which costs \(2\) symbols. Then, for \(i\in [2]\), it reads \(a_i\), \(c_i\), and \(L_{i}[1]=a_i+2c_i+b_i+d_i\). Using these downloaded symbols, it computes $P_i=\big(a_i+c_i,\; L_{i}[1]-(a_i+2c_i)\big)^\top, i\in[2]$. Therefore, \(P_1\) and \(P_2\) are recovered by reading \(2\cdot3\) symbols, and the coordinator obtains \(G_1^F=P_1+P_2\). Finally, the remaining final global parity node is obtained as \(G_2^F=G_1^I-G_1^F\). Therefore the total read bandwidth is \(\gamma_R=2+2\cdot 3=8\). For these parameters, the lower bound in Theorem~\ref{thm:final_thm} is also \(8\), hence this conversion attains the optimal read bandwidth.

\subsubsection{General Construction}

Set \(g=g^F=g^I\) and $\alpha=g+\ell$. We now give the general construction underlying the preceding example. The construction uses one MDS array block for each local group. The only repair property required from this block is that its \(g\) global parity nodes can be optimally repaired from its \(r+\ell\) information and local parity nodes. Universal optimal-repair MDS array codes, such as those guaranteed by Theorem~\ref{thm:existence}, would provide this property, but they are stronger than needed here. We therefore first construct a small-subpacketization MDS array block tailored to this single repair pattern.

\begin{lemma}[]
\label{lem:fixed_parity_repair1}
    Let \(r,\ell,g\) be positive integers and set \(\alpha=g+\ell\). Let \(\mathbb F_q\) be a finite field over which there exists a \((g+\ell)\times r\) superregular matrix, i.e., a matrix whose square submatrices are all nonsingular. Such a matrix can be obtained, for example, from a Cauchy matrix whenever \(q\ge r+g+\ell\). Then there exists a systematic linear \([r+\ell+g,r,\alpha]_q\) MDS array code with nodes
    \[
    X_1,\ldots,X_r,\quad L_1,\ldots,L_\ell,\quad P_1,\ldots,P_g
    \]
    such that \(P_1,\ldots,P_g\) can be optimally repaired from \(X_1,\ldots,X_r,L_1,\ldots,L_\ell\) by downloading the first \(g\) subsymbols from each helper node.
\end{lemma}

\begin{proof}
    Let \(\Gamma=(\gamma_{\rho,i})_{\rho\in[g+\ell],\,i\in[r]}\) be a \((g+\ell)\times r\) superregular matrix over \(\mathbb F_q\). For each \(i\in[r]\), let \(X_i=(x_{i,1},\ldots,x_{i,g+\ell})^\top\in\mathbb F_q^{g+\ell}\). For each \(\rho\in[g+\ell]\), define \(Q_\rho:=\sum_{i=1}^r\gamma_{\rho,i}X_i =(q_{\rho,1},\ldots,q_{\rho,g+\ell})^\top \in\mathbb F_q^{g+\ell}\) where \(q_{\rho,c}= \sum_{i=1}^{r}\gamma_{\rho,i}x_{i,c}, c\in[g+\ell].\) Let \(e_1,\ldots,e_{g+\ell}\) denote the standard basis vectors of \(\mathbb F_q^{g+\ell}\). Define the nodes \( L_b:=Q_b+\sum_{h=1}^g q_{\ell+h,g+b}e_h\), \(b\in[\ell]\) and $P_h:=Q_{\ell+h}, h\in[g]$.

    We first prove that the resulting array code is MDS. Without the piggyback terms, each coordinate induces an \([r+g+\ell,r]\) scalar MDS code, since the matrix \(\Gamma\) is superregular. The piggyback terms only modify the first \(g\) coordinates by adding symbols from the last \(\ell\) coordinates. Hence, given any \(r\) surviving nodes, we first decode the last \(\ell\) coordinates \(g+1,\ldots,g+\ell\). These coordinates are unaffected by the piggyback terms, and therefore the corresponding scalar MDS codes can be decoded directly. In particular, all symbols \(q_{\ell+h,g+b}\), \(h\in[g]\), \(b\in[\ell]\), are recovered.

    We then decode the first \(g\) coordinates. For each \(h\in[g]\), the piggyback terms in coordinate \(h\) are now known and can be subtracted from the surviving local parity symbols. After removing these terms, the remaining symbols are exactly those of the original \(h\)-th scalar MDS code. Hence the first \(g\) coordinates can also be decoded. Therefore, any \(r\) surviving nodes determine the whole message, and the resulting array code is MDS.

    It remains to verify the repair property. Let \(\pi:\mathbb F_q^{g+\ell}\to\mathbb F_q^g\) denote the projection onto the first \(g\) coordinates. To repair \(P_1,\ldots,P_g\), the decoder downloads \(\pi(X_i)\) from each information node \(X_i\) and \(\pi(L_b)\) from each node \(L_b\). From the downloaded information-node subsymbols, the decoder can compute \(\pi(Q_\rho)=(q_{\rho,1},\ldots,q_{\rho,g})^\top\) for all \(\rho\in[g+\ell]\). Since \(P_h=Q_{\ell+h}\), the first \(g\) coordinates of all global parity nodes are recovered.  Moreover, for each \(b\in[\ell]\),
    \[
        \pi(L_b)-\pi(Q_b) =(q_{\ell+1,g+b},\ldots,q_{\ell+g,g+b})^\top .
    \]
    Therefore, the \((g+b)\)-th coordinates of \(P_1,\ldots,P_g\) are recovered. Repeating this for every \(b\in[\ell]\) recovers all remaining coordinates of the global parity nodes. The total download is \(g(r+\ell)\), which equals the centralized repair cut-set bound
    \[
        \frac{g(r+\ell)}{g+(r+\ell)-r}\alpha=g(r+\ell),
    \]
    where \(\alpha=g+\ell\). Hence the repair achieves the cut-set bound and is optimal.
\end{proof}

We now construct the LRCC.  For each \(t\in[\lambda]\), the local groups of the \(t\)-th final codeword are indexed by \([\mu]^t\). For \(\tau\in[\mu]^t\), define $\iota(\tau):=\tau-(t-1)\mu\in[\mu]$, which gives the local index of the group \(\tau\) within the \(t\)-th final codeword.

Throughout this construction, each node is viewed as a column vector in \(\mathbb F_q^\alpha\). For each local group \(\tau\in[\lambda \mu]\), take an independent copy of the MDS array code in Lemma~\ref{lem:fixed_parity_repair1}. The \(r\) information nodes of this copy are \(X_{[r]^\tau}\), and the \(\ell\) local parity nodes are \(L_{[\ell]^\tau}\). This copy also produces \(g\) virtual global parity nodes \(P_{\tau,1},\ldots,P_{\tau,g}\in\mathbb F_q^\alpha\). Define the
stacked vector
\[
\mathbf P_\tau:=(P_{\tau,1}^{\top},\ldots,P_{\tau,g}^{\top})^{\top} \in\mathbb F_q^{g\alpha}.
\]

Now we choose mixing matrices from the block-scalar diagonal family. More precisely, for \(c_1,\ldots,c_g\in\mathbb F_q^*\), define $D(c_1,\ldots,c_g):=\operatorname{diag}(c_1I_\alpha,\ldots,c_gI_\alpha)\in\operatorname{GL}_{g\alpha}(\mathbb F_q).$ We choose $A_i=D(a_{i,1},\ldots,a_{i,g}), i\in[\mu]$ and $ B_t=D(b_{t,1},\ldots,b_{t,g}), t\in[\lambda],$ where all \(a_{i,h}\) and \(b_{t,h}\) are nonzero. These matrices are used to ensure that both the initial code and the final code are optimal-distance LRCs. The existence of such choices over sufficiently large finite fields follows from the mixing lemma in Appendix~\ref{appe}.

For the \(t\)-th final codeword, define the stacked vector of its \(g\) global parity nodes by
\[
    \mathbf G^F_{[g]^t}:= \sum_{\tau\in[\mu]^t} A_{\iota(\tau)}\mathbf P_\tau \in \mathbb F_q^{g\alpha}.
\]
Writing
\[
    \mathbf G^F_{[g]^t} = \big( (G^F_{(t-1)g+1})^\top,\ldots,(G^F_{tg})^\top \big)^\top, \qquad G^F_i\in\mathbb F_q^\alpha,
\]
the \(g\) consecutive \(\alpha\)-dimensional blocks \(G^F_i\), \(i\in[g]^t\), are stored as the \(g\) final global parity nodes of the \(t\)-th final codeword. Thus the \(t\)-th final codeword consists of the information nodes \(X_{[k^F]^t}\), the local parity nodes \(L_{[\mu\ell]^t}\), and the global parity nodes \(G^F_{[g]^t}\). All final codewords use the same matrices \(A_1,\ldots,A_\mu\), and hence are codewords of the same final LRC. By the choice of the matrices \(A_1,\ldots,A_\mu\), this final code is an optimal-distance \((k^F,g,r,\ell,\alpha)\)-LRC.

The initial codeword contains all \(\lambda \mu\) local groups. We define
the stacked vector of its \(g\) initial global parity nodes as
\[
    \mathbf G^I_{[g]} := \sum_{t=1}^{\lambda} B_t\mathbf G^F_{[g]^t} = \sum_{t=1}^{\lambda}\sum_{\tau\in[\mu]^t} B_tA_{\iota(\tau)}\mathbf P_\tau \in\mathbb F_q^{g\alpha}.
\]
Writing
\[
    \mathbf G^I_{[g]} = \big( (G^I_1)^\top,\ldots,(G^I_g)^\top \big)^\top, \qquad G^I_i\in\mathbb F_q^\alpha,
\]
the \(g\) consecutive \(\alpha\)-dimensional blocks \(G^I_1,\ldots,G^I_g\) are stored as the \(g\) initial global parity nodes. By the choice of the matrices \(A_1,\ldots,A_\mu\) and \(B_1,\ldots,B_\lambda\), the initial code is an optimal-distance \((k^I=\lambda k^F,g,r,\ell,\alpha)\)-LRC.

We next describe the conversion procedure. All information nodes and local parity nodes are preserved. The coordinator first downloads all \(g\) initial global parity nodes, equivalently the whole stacked vector \(\mathbf G^I_{[g]}\), which costs \(g\alpha\) symbols. Then, for each \(t\in[\lambda-1]\) and \(\tau\in[\mu]^t\), it uses Lemma~\ref{lem:fixed_parity_repair1} to recover \(\mathbf P_\tau\) from the preserved local nodes \(X_{[r]^\tau}\) and \(L_{[\ell]^\tau}\). This requires \(\frac{g(r+\ell)}{g+\ell}\alpha\) symbols for each local group. After obtaining all \(\mathbf P_\tau\) with \(\tau\in[\mu]^t\), the coordinator computes
\[
    \mathbf G^F_{[g]^t} = \sum_{\tau\in[\mu]^t} A_{\iota(\tau)}\mathbf P_\tau, \qquad t\in[\lambda-1].
\]
Finally, since \(B_\lambda\) is invertible, the stacked vector of the last
final global parity nodes is obtained as
\[
    \mathbf G^F_{[g]^\lambda} = B_\lambda^{-1} \left( \mathbf G^I_{[g]} - \sum_{t=1}^{\lambda-1}B_t\mathbf G^F_{[g]^t} \right).
\]
The vector \(\mathbf G^F_{[g]^\lambda}\) is then split into its \(g\) consecutive \(\alpha\)-dimensional blocks, which are stored as the global parity nodes \(G^F_{[g]^\lambda}\) of the last final codeword.

Therefore, the read bandwidth is
\[
    \gamma_R = g\alpha + (\lambda-1)\mu\frac{g(r+\ell)}{g+\ell}\alpha = g\alpha + (\lambda-1)\frac{g}{g+\ell}(k^F+\mu\ell)\alpha .
\]
When \(g^F=g^I=g\), the lower bound in Theorem~\ref{thm:final_thm} reduces to this quantity. Hence the construction attains the optimal read bandwidth in the case \(g^F=g^I\).

\subsection{\texorpdfstring{$g^F>g^I$}{gF > gI}}
\label{subsec5.2}
\subsubsection{An illustrative example}
Consider a global split conversion from an initial \((8,1,2,1,\alpha=3)\)-LRC codeword to two final \((4,2,2,1,\alpha=3)\)-LRC codewords. Here \(\lambda^F=2,\mu=2, g^I=1 \) and \(g^F=2\). We take \(\alpha=3\) and work over \(\mathbb F_5\).

Let the eight information nodes be \(X_{2i-1}=(a_i,b_i,c_i)^\top\) and \(X_{2i}=(d_i,e_i,f_i)^\top\), \(i\in[4]\). For each \(i\in[4]\), define \(A_i=a_i+d_i\), \(\bar A_i=b_i+e_i\), \(B_i=a_i+2d_i\), \(\bar B_i=b_i+2e_i\), \(C_i=c_i+f_i\) and \(D_i=c_i+2f_i\). In addition, define the local parity components \(L_i^{1}=a_i+3d_i+C_i\), \(L_i^{2}=b_i+3e_i+D_i\), and \(T_i=c_i+3f_i\). The local parity node of the \(i\)-th local group is \(L_i=(L_i^{1},L_i^{2},T_i)^\top\). The two virtual global parity nodes associated with this local group are \(P_{i,1}=(A_i,\bar A_i+B_i,C_i)^\top\) and \(P_{i,2}=(B_i,\bar B_i,D_i)^\top\).

The first final codeword has global parity nodes \(G^F_1=P_{1,1}+P_{2,1}\) and \(G^F_2=P_{1,2}+P_{2,2}\). The second final codeword has global parity nodes \(G^F_3=P_{3,1}+P_{4,1}\) and \(G^F_4=P_{3,2}+P_{4,2}\). The initial codeword has one global parity node \(G^I_1=G^F_1+G^F_3=\sum_{i=1}^4P_{i,1}\), namely \(G^I_1= \left( \sum_{i=1}^4 A_i,\, \sum_{i=1}^4(\bar A_i+B_i),\, \sum_{i=1}^4 C_i \right)^\top\).

The initial codeword and the two final codewords are shown in Table~\ref{tab:example_gf_gt_gi_two_layer}. The yellow entries in the initial codeword indicate the symbols read during conversion.

\begin{table*}[t]
\centering
\caption{The \((8,1,2,1,\alpha=3)\)-to-\((4,2,2,1,\alpha=3)\) split conversion example. Yellow entries are read during conversion.}
\label{tab:example_gf_gt_gi_two_layer}
\begingroup
\tinytabstyle
\setlength{\tabcolsep}{2.2pt}
\renewcommand{\arraystretch}{0.55}
\setlength{\fboxsep}{0.5pt}

\textbf{Initial codeword.}\par\vspace{0.2ex}
\begin{adjustbox}{max width=0.88\textwidth}
\begin{tabular}{c|c|c|c|c}
\hline
Pos. & Node & \(1\) & \(2\) & \(3\)\\
\hline
1  & \(X_1\)  & \readcell{\(a_1\)} & \readcell{\(b_1\)} & \(c_1\)\\
2  & \(X_2\)  & \readcell{\(d_1\)} & \readcell{\(e_1\)} & \(f_1\)\\
3  & \(L_1\)  & \readcell{\(L_1^{1}\)} & \readcell{\(L_1^{2}\)} & \(T_1\)\\
\hline
4  & \(X_3\)  & \readcell{\(a_2\)} & \readcell{\(b_2\)} & \(c_2\)\\
5  & \(X_4\)  & \readcell{\(d_2\)} & \readcell{\(e_2\)} & \(f_2\)\\
6  & \(L_2\)  & \readcell{\(L_2^{1}\)} & \readcell{\(L_2^{2}\)} & \(T_2\)\\
\hline
7  & \(X_5\)  & \(a_3\) & \readcell{\(b_3\)} & \(c_3\)\\
8  & \(X_6\)  & \(d_3\) & \readcell{\(e_3\)} & \(f_3\)\\
9  & \(L_3\)  & \(L_3^{1}\) & \readcell{\(L_3^{2}\)} & \(T_3\)\\
\hline
10 & \(X_7\)  & \(a_4\) & \readcell{\(b_4\)} & \(c_4\)\\
11 & \(X_8\)  & \(d_4\) & \readcell{\(e_4\)} & \(f_4\)\\
12 & \(L_4\)  & \(L_4^{1}\) & \readcell{\(L_4^{2}\)} & \(T_4\)\\
\hline
13 & \(G^I_1\)
& \readcell{\(\sum_{i=1}^4 A_i\)}
& \readcell{\(\sum_{i=1}^4(\bar A_i+B_i)\)}
& \readcell{\(\sum_{i=1}^4 C_i\)}\\
\hline
\end{tabular}
\end{adjustbox}

\vspace{0.8mm}

\textbf{Final codewords.}\par\vspace{0.2ex}
\begin{adjustbox}{max width=\textwidth}
\begin{tabular}{c|c|c|c|c||c|c|c|c|c}
\hline
Pos. & Node in \(C^F_1\) & \(1\) & \(2\) & \(3\)
& Pos. & Node in \(C^F_2\) & \(1\) & \(2\) & \(3\)\\
\hline
1 & \(X_1\) & \(a_1\) & \(b_1\) & \(c_1\)
& 1 & \(X_5\) & \(a_3\) & \(b_3\) & \(c_3\)\\
2 & \(X_2\) & \(d_1\) & \(e_1\) & \(f_1\)
& 2 & \(X_6\) & \(d_3\) & \(e_3\) & \(f_3\)\\
3 & \(L_1\) & \(L_1^{1}\) & \(L_1^{2}\) & \(T_1\)
& 3 & \(L_3\) & \(L_3^{1}\) & \(L_3^{2}\) & \(T_3\)\\
\hline
4 & \(X_3\) & \(a_2\) & \(b_2\) & \(c_2\)
& 4 & \(X_7\) & \(a_4\) & \(b_4\) & \(c_4\)\\
5 & \(X_4\) & \(d_2\) & \(e_2\) & \(f_2\)
& 5 & \(X_8\) & \(d_4\) & \(e_4\) & \(f_4\)\\
6 & \(L_2\) & \(L_2^{1}\) & \(L_2^{2}\) & \(T_2\)
& 6 & \(L_4\) & \(L_4^{1}\) & \(L_4^{2}\) & \(T_4\)\\
\hline
7 & \(G^F_1\) & \(A_1+A_2\) & \(\bar A_1+\bar A_2+B_1+B_2\) & \(C_1+C_2\) & 7 & \(G^F_3\) & \(A_3+A_4\) & \(\bar A_3+\bar A_4+B_3+B_4\) & \(C_3+C_4\)\\
8 & \(G^F_2\) & \(B_1+B_2\) & \(\bar B_1+\bar B_2\) & \(D_1+D_2\) & 8 & \(G^F_4\) & \(B_3+B_4\) & \(\bar B_3+\bar B_4\) & \(D_3+D_4\)\\
\hline
\end{tabular}
\end{adjustbox}

\endgroup
\end{table*}

A direct computation using Magma over \(\mathbb F_5\) gives
\(d(C^I)=3\) and \(d(C^F)=4\). Hence both the initial and final codes
are optimal-distance LRCs.

We now describe the conversion. First, the coordinator reads the whole initial global parity node \(G^I_1\), which costs \(3\) symbols. Then, for \(i\in[2]\), it reads \(a_i,b_i,d_i,e_i,L_i^{1},L_i^{2}\). From these symbols, it computes \(C_i=L_i^{1}-(a_i+3d_i)\) and \(D_i=L_i^{2}-(b_i+3e_i)\), and hence obtains both virtual global parity nodes \(P_{i,1}=(A_i,\bar A_i+B_i,C_i)^\top\) and \(P_{i,2}=(B_i,\bar B_i,D_i)^\top \). Therefore, it generates \(G^F_1=P_{1,1}+P_{2,1}\) and \(G^F_2=P_{1,2}+P_{2,2}\). This step reads \(12\) symbols.

Next, \(G^F_3=G^I_1-G^F_1\) is obtained. Write \(G^F_3=(A_{34},\bar A_{34}+B_{34},C_{34})^\top\), where \(A_{34}=A_3+A_4\), \(\bar A_{34}=\bar A_3+\bar A_4\), \(B_{34}=B_3+B_4\) and \(C_{34}=C_3+C_4\). For \(i=3,4\), the coordinator reads \(b_i,e_i,L_i^{2}\). It first computes \(\bar A_{34}=(b_3+e_3)+(b_4+e_4)\). Then, by subtracting \(\bar A_{34}\) from the second coordinate of \(G^F_3\), it obtains \(B_{34}=B_3+B_4\). Moreover, it computes \(\bar B_3+\bar B_4=(b_3+2e_3)+(b_4+2e_4)\) and \(D_3+D_4 = L_3^{2}+L_4^{2} -\bigl((b_3+3e_3)+(b_4+3e_4)\bigr)\). Thus \(G^F_4=(B_3+B_4,\bar B_3+\bar B_4,D_3+D_4)^\top\) is generated. This step reads \(6\) symbols.

Therefore, all final global parity nodes are generated. The total read bandwidth is $\gamma_R=3+12+6=21.$ For these parameters, the lower bound in Theorem~\ref{thm:final_thm} is also \(21\). Hence this conversion attains the optimal read bandwidth.

\subsubsection{General Construction}

Throughout this subsection, we assume \(g^I<g^F\le r\) and set $\alpha=g^F+\ell$. Although MDS array codes with universal optimal-repair properties, such as those guaranteed by Theorem~\ref{thm:existence}, may be used, only two prescribed repair patterns are needed here. First, all \(g^F\) global parity nodes must be optimally repaired from the \(r+\ell\) information and local parity nodes. Second, the last \(g^F-g^I\) global parity nodes must be optimally repaired from the same nodes when the first \(g^I\) global parity nodes are available as side information. We therefore construct a small-subpacketization MDS array code tailored to these two repair patterns, using one layer of local-parity piggybacks and one layer of head-global piggybacks.

\begin{lemma}
\label{lem:fixed_parity_repair_side}
  Let \(r,\ell,g^I,g^F\) be positive integers with \(g^I<g^F\le r\), and set \(\alpha=g^F+\ell\). Let \(\mathbb F_q\) be a finite field over which there exists a \((g^F+\ell)\times r\) superregular matrix. For example, a \((g^F+\ell)\times r\) Cauchy matrix over a field with \(q\ge r+g^F+\ell\) is superregular and satisfies the required condition. Then there exists a systematic linear \([r+\ell+g^F,r,\alpha]_q\) MDS array code with nodes
  \[
  X_1,\ldots,X_r,\quad L_1,\ldots,L_\ell,\quad P_1,\ldots,P_{g^F}
  \]
  such that \(P_1,\ldots,P_{g^F}\) can be optimally repaired from \(X_1,\ldots,X_r,L_1,\ldots,L_\ell\) by downloading the first \(g^F\) subsymbols from each helper node. Moreover, \(P_{g^I+1},\ldots,P_{g^F}\) can be optimally repaired from \(X_1,\ldots,X_r,L_1,\ldots,L_\ell,P_1,\ldots,P_{g^I}\) by downloading the \(g^F-g^I\) subsymbols in coordinate positions \(g^I+1,\ldots,g^F\) from each helper node. In particular, when \(P_1,\ldots,P_{g^I}\) are available as side information, \(P_{g^I+1},\ldots,P_{g^F}\) can be recovered by downloading only these \(g^F-g^I\) subsymbols from each of \(X_1,\ldots,X_r,L_1,\ldots,L_\ell\).
\end{lemma}

\begin{proof}
    Let \(\Gamma=(\gamma_{\rho,i})_{\rho\in[g^F+\ell],\,i\in[r]}\) be a \((g^F+\ell)\times r\) superregular matrix over \(\mathbb F_q\). For each \(i\in[r]\), let \(X_i=(x_{i,1},\ldots,x_{i,\alpha})^\top\in\mathbb F_q^\alpha\). For each \(\rho\in[g^F+\ell]\), define \(Q_\rho:=\sum_{i=1}^r\gamma_{\rho,i}X_i =(q_{\rho,1},\ldots,q_{\rho,\alpha})^\top\in\mathbb F_q^\alpha\) where \(q_{\rho,c}=\sum_{i=1}^r\gamma_{\rho,i}x_{i,c}\) for \(c\in[\alpha]\).

    Let \(e_1,\ldots,e_\alpha\) denote the standard basis vectors of \(\mathbb F_q^\alpha\). Define $L_b:=Q_b+\sum_{h=1}^{g^F}q_{\ell+h,g^F+b}e_h,  b\in[\ell].$ For \(j=g^I+1,\ldots,g^F\), define \(P_j:=Q_{\ell+j}\); for \(h\in[g^I]\), define \(P_h:=Q_{\ell+h}+\sum_{j=g^I+1}^{g^F}q_{\ell+j,h}e_j\). Thus, for \(h\in[g^I]\) and \(j=g^I+1,\ldots,g^F\), $P_h[j]=q_{\ell+h,j}+q_{\ell+j,h}$, while \(P_h[c]=q_{\ell+h,c}\) for all other coordinates \(c\). Together with the systematic nodes \(X_1,\ldots,X_r\), the nodes \(L_1,\ldots,L_\ell\) and \(P_1,\ldots,P_{g^F}\) define a systematic linear array code.

    We first prove that the resulting array code is MDS. Before adding the piggyback terms, each coordinate induces an \([r+\ell+g^F,r]\) scalar MDS code, since \(\Gamma\) is superregular. The piggybacks are triangular with respect to the coordinate order
    \[
    g^F+1,\ldots,g^F+\ell,\quad
    1,\ldots,g^I,\quad
    g^I+1,\ldots,g^F .
    \]
    Indeed, the local-parity piggybacks in coordinates \(1,\ldots,g^F\) depend only on the last \(\ell\) coordinates, while the head-global piggybacks in coordinates \(g^I+1,\ldots,g^F\) depend only on coordinates \(1,\ldots,g^I\).

    Fix any set of \(r\) surviving nodes. We first decode the last \(\ell\) coordinates \(g^F+1,\ldots,g^F+\ell\). These coordinates contain no piggyback terms, so each of them is decoded directly as a scalar MDS instance. In particular, all symbols \(q_{\ell+h,g^F+b}\), \(h\in[g^F]\), \(b\in[\ell]\), are known.

    We then decode coordinates \(1,\ldots,g^I\). After subtracting the known local-parity piggybacks from the surviving local parity subsymbols, each of these coordinates again reduces to an unpiggybacked scalar MDS instance. Hence all information subsymbols in coordinates \(1,\ldots,g^I\) are recovered. In particular, all symbols \(q_{\ell+j,h}\), \(j=g^I+1,\ldots,g^F\), \(h\in[g^I]\), are known.

    Finally, consider a coordinate \(j\in\{g^I+1,\ldots,g^F\}\). The local-parity piggybacks are known from the last \(\ell\) coordinates, and the head-global piggybacks are known from coordinates \(1,\ldots,g^I\). After subtracting these terms, the surviving subsymbols in coordinate \(j\) are exactly symbols of the unpiggybacked \(j\)-th scalar MDS instance. Thus coordinates \(g^I+1,\ldots,g^F\) are also decoded.

    Therefore any \(r\) surviving nodes determine the whole message, and the array code is MDS.

    It remains to verify the two repair properties. Let \(\pi:\mathbb F_q^\alpha\to\mathbb F_q^{g^F}\) be the projection onto the first \(g^F\) coordinates. For the first repair, the decoder downloads \(\pi(X_i)\) from each \(X_i\) and \(\pi(L_b)\) from each \(L_b\). From the downloaded subsymbols of the systematic nodes, it computes \(\pi(Q_\rho)=(q_{\rho,1},\ldots,q_{\rho,g^F})^\top\) for every \(\rho\in[g^F+\ell]\). Moreover, for each \(b\in[\ell]\),
    \[
    \pi(L_b)-\pi(Q_b)=(q_{\ell+1,g^F+b},\ldots,q_{\ell+g^F,g^F+b})^\top.
    \]
    Thus the \((g^F+b)\)-th coordinates of \(Q_{\ell+1},\ldots,Q_{\ell+g^F}\) are recovered. As \(b\) ranges over \([\ell]\), all remaining coordinates of these vectors are recovered. Therefore \(Q_{\ell+1},\ldots,Q_{\ell+g^F}\), and hence \(P_1,\ldots,P_{g^F}\), are completely determined. The total download is \(g^F(r+\ell)\) symbols. Since \(\alpha=g^F+\ell\), the centralized repair cut-set bound for repairing \(g^F\) nodes from \(r+\ell\) helper nodes is
    \[
    \frac{g^F(r+\ell)}{g^F+(r+\ell)-r}\alpha = g^F(r+\ell),
    \]
    where \(\alpha=g^F+\ell\). Hence the first repair is optimal.
 
    For the second repair, let $\pi':\mathbb F_q^\alpha\to\mathbb F_q^{g^F-g^I}$ be the projection onto coordinate positions \(g^I+1,\ldots,g^F\). The decoder downloads \(\pi'(X_i)\) from each \(X_i\), \(\pi'(L_b)\) from each \(L_b\), and \(\pi'(P_h)\) from each \(P_h\), \(h\in[g^I]\). From the downloaded subsymbols of the systematic nodes, the decoder computes \(q_{\rho,c}\) for every \(\rho\in[g^F+\ell]\) and \(c\in\{g^I+1,\ldots,g^F\}\). Now fix \(j\in\{g^I+1,\ldots,g^F\}\). The subsymbols of \(P_j=Q_{\ell+j}\) in positions \(g^I+1,\ldots,g^F\) are already known from the systematic downloads.

    For each \(b\in[\ell]\), the downloaded local parity subsymbol satisfies \(L_b[j]=q_{b,j}+q_{\ell+j,g^F+b}.\) Since \(q_{b,j}\) is computable from the systematic downloads, this determines \(q_{\ell+j,g^F+b}\). Hence the last \(\ell\) coordinates of \(Q_{\ell+j}\) are known. Moreover, for every \(h\in[g^I]\), the downloaded subsymbol of \(P_h\) satisfies \(P_h[j]=q_{\ell+h,j}+q_{\ell+j,h}\). Since \(q_{\ell+h,j}\) is computable from the systematic downloads, this determines \(q_{\ell+j,h}\). Thus the first \(g^I\) coordinates of \(Q_{\ell+j}\) are also known. It follows that \(P_j=Q_{\ell+j}\) is completely recovered. Repeating this argument for every \(j\in\{g^I+1,\ldots,g^F\}\) recovers \(P_{g^I+1},\ldots,P_{g^F}\).

    The second repair downloads \(g^F-g^I\) subsymbols from each of \(r+\ell+g^I\) helper nodes, for a total of \((g^F-g^I)(r+\ell+g^I)\) symbols. Since \(\alpha=g^F+\ell\), the centralized repair cut-set bound is
    \[
    \frac{(g^F-g^I)(r+\ell+g^I)} {(g^F-g^I)+(r+\ell+g^I)-r}\alpha = (g^F-g^I)(r+\ell+g^I).
    \]
    Hence the second repair is also optimal. If \(P_1,\ldots,P_{g^I}\) are already available as side information, then their \(\pi'\)-projections need not be downloaded, and the same recovery uses downloads only from \(X_1,\ldots,X_r,L_1,\ldots,L_\ell\).
\end{proof}

We now construct the LRCC. For each \(t\in[\lambda]\), the local groups of the \(t\)-th final codeword are indexed by \([\mu]^t\). For \(\tau\in[\mu]^t\), define \(\iota(\tau):=\tau-(t-1)\mu\in[\mu]\), which denotes the local index of the group \(\tau\) inside the \(t\)-th final codeword.

Throughout the construction, each node is viewed as a column vector in \(\mathbb F_q^\alpha\). For every local group \(\tau\in[\lambda\mu]\), take an independent copy of the underlying array block used in Lemma~\ref{lem:fixed_parity_repair_side}. Its information nodes and local parity nodes are \(X_{[r]^\tau}\) and \(L_{[\ell]^\tau}\), respectively. The local parity nodes already contain the local-parity piggybacks. For \(h\in[g^F]\), let \(P_{\tau,h}\in\mathbb F_q^\alpha\) denote the body virtual global parity vector corresponding to \(Q_{\ell+h}\) in that copy. The head-global piggybacks are not included in \(P_{\tau,h}\); they will be introduced after mixing the local groups. Write
\[
\mathbf P_\tau := \bigl( P_{\tau,1}^\top,\ldots,P_{\tau,g^F}^\top \bigr)^\top \in\mathbb F_q^{g^F\alpha}.
\]

We next choose the mixing matrices. Let \(\mathcal C_0\) denote the unpiggybacked MDS array code underlying Lemma~\ref{lem:fixed_parity_repair_side} whose local parity nodes are \(Q_1,\ldots,Q_\ell\) and whose virtual global parity nodes are \(Q_{\ell+1},\ldots,Q_{\ell+g^F}\). For \(c_1,\ldots,c_{g^F}\in\mathbb F_q^*\), define
\[
D(c_1,\ldots,c_{g^F}) := \operatorname{diag} \bigl( c_1I_\alpha,\ldots,c_{g^F}I_\alpha \bigr).
\]

To describe the mixing step compactly, we use the following notation. Given \(N\) copies \(\mathcal C_0^{(1)},\ldots,\mathcal C_0^{(N)}\) of \(\mathcal C_0\), each encoding a disjoint collection of \(r\) message blocks, let $\mathbf P^{(i)}\in\mathbb F_q^{g^F\alpha}$ denote the stacked vector of the \(g^F\) virtual global parity nodes of the \(i\)-th copy. For block-scalar diagonal matrices \(M_1,\ldots,M_N\), let $\mathcal L_{\mathcal C_0}(M_1,\ldots,M_N)$ denote the systematic \((Nr,g^F,r,\ell,\alpha)\)-LRC whose \(N\) local groups are formed by the information and local parity nodes of these copies and whose stacked global parity vector is $\sum_{i=1}^N M_i\mathbf P^{(i)}.$
The formal definition of this auxiliary construction is given in
Appendix~\ref{appe}.

Applying Lemma~\ref{lem:mixing_auxiliary_lrc} to \(\mathcal C_0\) with the parameter \(g\) in Appendix~\ref{appe} instantiated as \(g^F\), over a sufficiently large finite field, we can choose
\[
A_i = D(a_{i,1},\ldots,a_{i,g^F}), \quad i\in[\mu], \qquad B_t = D(b_{t,1},\ldots,b_{t,g^F}),
\quad t\in[\lambda],
\]
where all \(a_{i,h}\) and \(b_{t,h}\) are nonzero, such that \(\mathcal L_{\mathcal{C}_0}(A_1,\ldots,A_\mu)\) and $\mathcal L_{\mathcal{C}_0} \bigl( B_tA_i:t\in[\lambda],\,i\in[\mu] \bigr)$ are optimal-distance LRCs. 

We first define the final code. For each \(t\in[\lambda]\), let
\[
\mathbf Y_t := \sum_{\tau\in[\mu]^t} A_{\iota(\tau)}\mathbf P_\tau = \bigl( Y_{t,1}^\top,\ldots,Y_{t,g^F}^\top \bigr)^\top \in\mathbb F_q^{g^F\alpha},
\]
where \(Y_{t,h}\in\mathbb F_q^\alpha\). For
\(h\in[g^F]\), define the corresponding final global parity node as
\[
G^F_{(t-1)g^F+h} :=
\begin{cases}
\displaystyle
Y_{t,h} + \sum_{j=g^I+1}^{g^F} Y_{t,j}[h]e_j, & h\in[g^I],\\[2mm]
Y_{t,h}, & h=g^I+1,\ldots,g^F,
\end{cases}
\]
where \(e_j\) is the \(j\)-th standard basis vector of \(\mathbb F_q^\alpha\).

For \(h\in[g^I]\), the preceding definition can equivalently be written as
\[
G^F_{(t-1)g^F+h} = \sum_{\tau\in[\mu]^t} a_{\iota(\tau),h} \left( P_{\tau,h} + \sum_{j=g^I+1}^{g^F} \frac{a_{\iota(\tau),j}}      {a_{\iota(\tau),h}} P_{\tau,j}[h]e_j \right).
\]
Indeed, before the \(h\)-th virtual global parity vector of a local group is multiplied by \(a_{\iota(\tau),h}\), the associated piggyback term from \(P_{\tau,j}[h]\) is assigned the coefficient \(a_{\iota(\tau),j}/a_{\iota(\tau),h}\). After multiplication, this term has coefficient \(a_{\iota(\tau),j}\), which is exactly the coefficient of \(P_{\tau,j}\) in \(Y_{t,j}\). Consequently,
\[
G^F_{(t-1)g^F+h}[j] = Y_{t,h}[j]+Y_{t,j}[h], \qquad h\in[g^I],\quad j=g^I+1,\ldots,g^F.
\]
This is precisely the coefficient-compatible head-global piggyback
required in the conversion.

The \(t\)-th final codeword consists of the information nodes \(X_{[k^F]^t}\), the local parity nodes \(L_{[\mu\ell]^t}\), and the global parity nodes \(G^F_{[g^F]^t}\). Since the same matrices \(A_1,\ldots,A_\mu\) and the same coefficient-compatible piggyback rule are used for every \(t\), all final codewords are codewords of the same final LRC.

We next define the initial code. It contains all \(\lambda \mu\) local groups and \(g^I\) global parity nodes. For \(h\in[g^I]\), define
\[
G^I_h := \sum_{t=1}^{\lambda} b_{t,h}G^F_{(t-1)g^F+h}.
\]
After removing the head-global piggybacks, the body part of \(G^I_h\) is \(\sum_{t=1}^{\lambda}b_{t,h}Y_{t,h}\), whose coefficient on \(P_{\tau,h}\), for \(\tau\in[\mu]^t\), is \(b_{t,h}a_{\iota(\tau),h}\). Hence the \(g^I\) initial body global parity nodes are precisely the first \(g^I\) global parity nodes obtained by puncturing \(\mathcal L_{\mathcal{C}_0}(B_tA_i:t\in[\lambda],\,i\in[\mu])\). By construction, all information nodes and local parity nodes are preserved during conversion, and therefore the conversion is stable.

We now verify that the initial and final codes are optimal-distance LRCs. Before adding the piggyback layers, the final body code is \(\mathcal L_{\mathcal{C}_0}(A_1,\ldots,A_\mu)\), and the initial body code with \(g^F\) global parity nodes is \(\mathcal L_{\mathcal{C}_0}(B_tA_i:t\in[\lambda],\,i\in[\mu])\). Both body codes are optimal-distance by Lemma~\ref{lem:mixing_auxiliary_lrc}. Puncturing the last \(g^F-g^I\) global parity nodes from the latter gives a body initial code whose minimum distance is at least $g^F+\ell+1-(g^F-g^I)=g^I+\ell+1.$ The Singleton-type bound for a \(( k^I,g^I,r,\ell,\alpha)\)-LRC gives the reverse inequality. Hence the punctured initial body code is also optimal-distance.

It remains to show that the two piggyback layers do not change the erasure-correcting capability. The local-parity piggybacks add symbols from coordinate \(g^F+b\) to coordinate \(h\), where \(b\in[\ell]\) and \(h\in[g^F]\). The head-global piggybacks add symbols depending on coordinates \(h\in[g^I]\) to coordinates \(j\in\{g^I+1,\ldots,g^F\}\). Therefore, all dependencies are acyclic under the coordinate order $g^F+1,\ldots,g^F+\ell,\quad 1,\ldots,g^I,\quad g^I+1,\ldots,g^F.$ Since the unpiggybacked code \(\mathcal C_0\) has no cross-coordinate coupling and the matrices \(A_i\) and \(B_t\) are block-scalar, any correctable node-erasure pattern can be decoded in this order. Namely, one first decodes coordinates \(g^F+1,\ldots,g^F+\ell\), then subtracts the known local-parity piggybacks and decodes coordinates \(1,\ldots,g^I\), and finally subtracts both types of piggybacks before decoding coordinates \(g^I+1,\ldots,g^F\). Hence the piggyback layers do not change the erasure-correcting capability of the corresponding body codes. Consequently, the final code is an optimal-distance \((k^F,g^F,r,\ell,\alpha)\)-LRC, and the initial code is an optimal-distance \((k^I,g^I,r,\ell,\alpha)\)-LRC.

We now describe the conversion procedure. The coordinator first downloads all initial global parity nodes \(G^I_1,\ldots,G^I_{g^I}\), which costs \(g^I\alpha\) symbols.

Next, for each \(t\in[\lambda-1]\) and every \(\tau\in[\mu]^t\), the coordinator downloads the first \(g^F\) subsymbols from every node in \(X_{[r]^\tau},L_{[\ell]^\tau}\). By the repair property in Lemma~\ref{lem:fixed_parity_repair_side}, these downloads recover all body virtual global parity vectors \(P_{\tau,1},\ldots,P_{\tau,g^F}\). The coordinator can therefore compute \(\mathbf Y_t\) and generate all final global parity nodes \(G^F_{[g^F]^t}\). The read cost of this step is
\[
(\lambda-1)\mu g^F(r+\ell) = (\lambda-1) \frac{g^F}{g^F+\ell} (k^F+\mu\ell)\alpha.
\]

It remains to generate the last final codeword. For every \(h\in[g^I]\), since \(b_{\lambda,h}\ne0\), the coordinator obtains
\[
G^F_{(\lambda-1)g^F+h} = b_{\lambda,h}^{-1} \left( G^I_h - \sum_{t=1}^{\lambda-1} b_{t,h}G^F_{(t-1)g^F+h} \right).
\]
Thus, the first \(g^I\) global parity nodes of the last final codeword are recovered.

To recover the remaining global parity nodes, the coordinator downloads from each node in \(X_{[k^F]^\lambda},L_{[\mu\ell]^\lambda}\) the \(g^F-g^I\) subsymbols in coordinate positions \(g^I+1,\ldots,g^F\). Fix \(j\in\{g^I+1,\ldots,g^F\}\). The downloaded information-node subsymbols determine \(Y_{\lambda,j}[c]\) for \(c=g^I+1,\ldots,g^F\), and also determine \(Y_{\lambda,h}[j]\) for every \(h\in[g^I]\). By the local-parity piggyback relation in Lemma~\ref{lem:fixed_parity_repair_side}, the downloaded local parity subsymbols determine \(Y_{\lambda,j}[g^F+1],\ldots,Y_{\lambda,j}[g^F+\ell]\). Finally, for every \(h\in[g^I]\), the already recovered head-global parity node satisfies
\[
G^F_{(\lambda-1)g^F+h}[j] = Y_{\lambda,h}[j]+Y_{\lambda,j}[h],
\]
which determines \(Y_{\lambda,j}[h]\). Therefore, the whole vector \(Y_{\lambda,j}\), and hence \(G^F_{(\lambda-1)g^F+j}=Y_{\lambda,j}\), is recovered. Repeating this argument for every \(j\in\{g^I+1,\ldots,g^F\}\) generates all remaining final global parity nodes. The read cost of this step is
\[
(g^F-g^I)(k^F+\mu\ell) = \frac{g^F-g^I}{g^F+\ell} (k^F+\mu\ell)\alpha.
\]

Hence the total read bandwidth is
\[
\begin{aligned}
\gamma_R &= g^I\alpha + (\lambda-1) \frac{g^F}{g^F+\ell} (k^F+\mu\ell)\alpha + \frac{g^F-g^I}{g^F+\ell} (k^F+\mu\ell)\alpha\\
&= \lambda \frac{g^F}{g^F+\ell} (k^F+\mu\ell)\alpha - g^I \left( \frac{k^F+\mu\ell}{g^F+\ell}-1 \right)\alpha.
\end{aligned}
\]
Since \(\lambda=\lambda^F\), this is exactly the \(g^F>g^I\) lower bound in Theorem~\ref{thm:final_thm}. Therefore the construction attains the optimal read bandwidth.

\subsection{\texorpdfstring{$g^F<g^I$}{gF < gI}}
\label{subsec5.3}
\subsubsection{An illustrative example}
Consider a global split conversion from an initial \((8,2,2,1,\alpha=5)\)-LRC codeword to two final \((4,1,2,1,\alpha=5)\)-LRC codewords. Here \(\lambda^F=2\), \(\mu=2\), \(g^I=2\), and \(g^F=1\). We take \(\alpha=5\) and work over \(\mathbb F_5\). Let \(\zeta=1\), \(\beta=2\), and \(\gamma=3\). The five subsymbol coordinates are partitioned into two prefix blocks, \(\mathcal{B}_1=(1,2)\) and \(\mathcal{B}_2=(3,4)\), together with a tail block \(\mathcal{T}=(5)\).

We use a logical cyclic ordering of the two prefix blocks in the final codewords, while keeping the tail block last. For the first final codeword, the logical block order is \(\mathcal B_1,\mathcal B_2,\mathcal T\), whereas for the second final codeword, it is \(\mathcal B_2,\mathcal B_1,\mathcal T\). Equivalently, define $\Pi_1(x_1,x_2,x_3,x_4,x_5)^\top =(x_1,x_2,x_3,x_4,x_5)^\top $ and $\Pi_2(x_1,x_2,x_3,x_4,x_5)^\top =(x_3,x_4,x_1,x_2,x_5)^\top.$ For every node in the \(t\)-th final codeword, we write its subsymbol vector in the logical coordinate order induced by \(\Pi_t\). This is only a relabeling of subsymbol coordinates: the physical storage order remains unchanged, and no preserved information or local parity node is moved or rewritten.

Let the eight information nodes be \(X_{2i-1}=(a_{i1},b_{i1},a_{i2},b_{i2},c_i)^\top\) and \(X_{2i}=(d_{i1},e_{i1},d_{i2},e_{i2},f_i)^\top\), for \(i\in[4]\). Also, for \(i\in[4]\) and \(s\in[2]\), define $A_{is}=a_{is}+d_{is}, \bar A_{is}=b_{is}+e_{is}$, $L^1_{is}=a_{is}+2d_{is}+b_{is}+e_{is}, L^2_{is}=b_{is}+2e_{is}$, $B_{is}=a_{is}+3d_{is}$ and $ \bar B_{is}=b_{is}+3e_{is}$. The local parity node of the \(i\)-th local group is $L_i=(L^1_{i1},L^2_{i1},L^1_{i2},L^2_{i2},T_i)^\top$, where $T_i= \begin{cases} 
c_i+\zeta f_i, & i\in\{1,3\},\\
\zeta^2c_i+\zeta^3f_i, & i\in\{2,4\}.
\end{cases}$

Define $R_1=c_1+\beta f_1+\beta^2c_2+\beta^3f_2$, $R_2=c_3+\beta f_3+\beta^2c_4+\beta^3f_4$, and for \(\theta\in\mathbb F_5^\ast\), $Z(\theta)= c_1+\theta f_1+\theta^2c_2+\theta^3f_2
+\theta^4c_3+\theta^5f_3+\theta^6c_4+\theta^7f_4$.

In physical coordinate order, the first final global parity node is
\[
G_1^F= \bigl( A_{11}+A_{21},\, \bar A_{11}+\bar A_{21},\, A_{12}+A_{22},\, \bar A_{12}+\bar A_{22},\, R_1+B_{11}+B_{21} \bigr)^\top.
\]
The second final global parity node is physically stored as
\[
G_{2,\mathrm{phy}}^F= \bigl( A_{31}+A_{41},\, \bar A_{31}+\bar A_{41},\, A_{32}+A_{42},\, \bar A_{32}+\bar A_{42},\, R_2+B_{32}+B_{42} \bigr)^\top.
\]
However, the second final codeword is described in the logical coordinate order \(\Pi_2\). Hence its logical final global parity node is
\[
G_2^F=\Pi_2G_{2,\mathrm{phy}}^F = \bigl( A_{32}+A_{42},\, \bar A_{32}+\bar A_{42},\, A_{31}+A_{41},\, \bar A_{31}+\bar A_{41},\, R_2+B_{32}+B_{42} \bigr)^\top.
\]
Under this logical ordering, both final codewords follow the same encoding rule: their local parity nodes use the same componentwise construction, and in each final global parity node the tail piggyback is attached to the first logical prefix block. Hence the two final codewords are codewords of the same final LRC.

The two initial global parity nodes are
\[
G_1^I= \biggl( \sum_{i=1}^4 A_{i1},\, \sum_{i=1}^4 \bar A_{i1},\, \sum_{i=1}^4 A_{i2},\, \sum_{i=1}^4 \bar A_{i2},\, Z(\beta) \biggr)^\top,
\]
and
\[
G_2^I= \biggl( \sum_{i=1}^4 B_{i1}+R_1,\, \sum_{i=1}^4 \bar B_{i1},\, \sum_{i=1}^4 B_{i2}+R_2,\, \sum_{i=1}^4 \bar B_{i2},\, Z(\gamma) \biggr)^\top.
\]

The initial codeword and the two final codewords are shown in Table~\ref{tab:example_gi_gt_gf_initial}. The yellow entries in the initial codeword are the symbols read during conversion. The final codewords are displayed in logical coordinate order; for the second final codeword, this means that coordinates \(1,2\) and \(3,4\) are swapped relative to the physical storage order.

\begin{table*}[t]
\centering
\caption{The \((8,2,2,1,\alpha=5)\)-to-\((4,1,2,1,\alpha=5)\) split conversion example. Yellow entries are read during conversion.}
\label{tab:example_gi_gt_gf_initial}
\begingroup
\tinytabstyle

\textbf{Initial codeword.}\par\vspace{0.2ex}
\begin{adjustbox}{max width=\textwidth}
\begin{tabular}{c|c|c|c|c|c|c}
\hline
Pos. & Node & \(1\) & \(2\) & \(3\) & \(4\) & \(5\)\\
\hline
1  & \(X_1\) & \(a_{11}\) & \(b_{11}\) & \readcell{\(a_{12}\)} & \(b_{12}\) & \(c_1\)\\
2  & \(X_2\) & \(d_{11}\) & \(e_{11}\) & \readcell{\(d_{12}\)} & \(e_{12}\) & \(f_1\)\\
3  & \(L_1\) & \(L^{1}_{11}\) & \(L^{2}_{11}\) & \readcell{\(L^{1}_{12}\)} & \(L^{2}_{12}\) & \(T_1\)\\
\hline
4  & \(X_3\) & \(a_{21}\) & \(b_{21}\) & \readcell{\(a_{22}\)} & \(b_{22}\) & \(c_2\)\\
5  & \(X_4\) & \(d_{21}\) & \(e_{21}\) & \readcell{\(d_{22}\)} & \(e_{22}\) & \(f_2\)\\
6  & \(L_2\) & \(L^{1}_{21}\) & \(L^{2}_{21}\) & \readcell{\(L^{1}_{22}\)} & \(L^{2}_{22}\) & \(T_2\)\\
\hline
7  & \(X_5\) & \readcell{\(a_{31}\)} & \(b_{31}\) & \(a_{32}\) & \(b_{32}\) & \(c_3\)\\
8  & \(X_6\) & \readcell{\(d_{31}\)} & \(e_{31}\) & \(d_{32}\) & \(e_{32}\) & \(f_3\)\\
9  & \(L_3\) & \readcell{\(L^{1}_{31}\)} & \(L^{2}_{31}\) & \(L^{1}_{32}\) & \(L^{2}_{32}\) & \(T_3\)\\
\hline
10 & \(X_7\) & \readcell{\(a_{41}\)} & \(b_{41}\) & \(a_{42}\) & \(b_{42}\) & \(c_4\)\\
11 & \(X_8\) & \readcell{\(d_{41}\)} & \(e_{41}\) & \(d_{42}\) & \(e_{42}\) & \(f_4\)\\
12 & \(L_4\) & \readcell{\(L^{1}_{41}\)} & \(L^{2}_{41}\) & \(L^{1}_{42}\) & \(L^{2}_{42}\) & \(T_4\)\\
\hline
13 & \(G^I_1\)
& \readcell{\(\sum_{i=1}^4 A_{i1}\)}
& \readcell{\(\sum_{i=1}^4 \bar{A}_{i1}\)}
& \readcell{\(\sum_{i=1}^4 A_{i2}\)}
& \readcell{\(\sum_{i=1}^4 \bar{A}_{i2}\)}
& \(Z(\beta)\)\\
14 & \(G^I_2\)
& \readcell{\(\sum_{i=1}^4 B_{i1}+R_1\)}
& \(\sum_{i=1}^4 \bar{B}_{i1}\)
& \readcell{\(\sum_{i=1}^4 B_{i2}+R_2\)}
& \(\sum_{i=1}^4 \bar{B}_{i2}\)
& \(Z(\gamma)\)\\
\hline
\end{tabular}
\end{adjustbox}

\vspace{0.8mm}

\textbf{Final codewords in logical coordinate order.}\par\vspace{0.2ex}
\begin{adjustbox}{max width=\textwidth}
\begin{tabular}{c|c|c|c|c|c|c||c|c|c|c|c|c|c}
\hline
Pos. & Node in \(C^F_1\) & \(1\) & \(2\) & \(3\) & \(4\) & \(5\)
& Pos. & Node in \(C^F_2\) & \(1\) & \(2\) & \(3\) & \(4\) & \(5\)\\
\hline
1 & \(X_1\) & \(a_{11}\) & \(b_{11}\) & \(a_{12}\) & \(b_{12}\) & \(c_1\)
& 1 & \(X_5\) & \(a_{32}\) & \(b_{32}\) & \(a_{31}\) & \(b_{31}\) & \(c_3\)\\
2 & \(X_2\) & \(d_{11}\) & \(e_{11}\) & \(d_{12}\) & \(e_{12}\) & \(f_1\)
& 2 & \(X_6\) & \(d_{32}\) & \(e_{32}\) & \(d_{31}\) & \(e_{31}\) & \(f_3\)\\
3 & \(L_1\) & \(L^{1}_{11}\) & \(L^{2}_{11}\) & \(L^{1}_{12}\) & \(L^{2}_{12}\) & \(T_1\)
& 3 & \(L_3\) & \(L^{1}_{32}\) & \(L^{2}_{32}\) & \(L^{1}_{31}\) & \(L^{2}_{31}\) & \(T_3\)\\
\hline
4 & \(X_3\) & \(a_{21}\) & \(b_{21}\) & \(a_{22}\) & \(b_{22}\) & \(c_2\)
& 4 & \(X_7\) & \(a_{42}\) & \(b_{42}\) & \(a_{41}\) & \(b_{41}\) & \(c_4\)\\
5 & \(X_4\) & \(d_{21}\) & \(e_{21}\) & \(d_{22}\) & \(e_{22}\) & \(f_2\)
& 5 & \(X_8\) & \(d_{42}\) & \(e_{42}\) & \(d_{41}\) & \(e_{41}\) & \(f_4\)\\
6 & \(L_2\) & \(L^{1}_{21}\) & \(L^{2}_{21}\) & \(L^{1}_{22}\) & \(L^{2}_{22}\) & \(T_2\)
& 6 & \(L_4\) & \(L^{1}_{42}\) & \(L^{2}_{42}\) & \(L^{1}_{41}\) & \(L^{2}_{41}\) & \(T_4\)\\
\hline
7 & \(G^F_1\)
& \(A_{11}+A_{21}\)
& \(\bar{A}_{11}+\bar{A}_{21}\)
& \(A_{12}+A_{22}\)
& \(\bar{A}_{12}+\bar{A}_{22}\)
& \(R_1+B_{11}+B_{21}\)
& 7 & \(G^F_2\)
& \(A_{32}+A_{42}\)
& \(\bar{A}_{32}+\bar{A}_{42}\)
& \(A_{31}+A_{41}\)
& \(\bar{A}_{31}+\bar{A}_{41}\)
& \(R_2+B_{32}+B_{42}\)\\
\hline
\end{tabular}
\end{adjustbox}

\endgroup
\end{table*}

To generate \(G_1^F\), the coordinator reads the third physical coordinate
from each of \(X_1,X_2,L_1,X_3,X_4,L_2\). Equivalently, for each
\(i\in[2]\), it downloads \(a_{i2}\), \(d_{i2}\), and \(L^1_{i2}\).
By definition, $A_{i2}=a_{i2}+d_{i2}, \bar A_{i2}=L^1_{i2}-(a_{i2}+2d_{i2})$ and $B_{i2}=a_{i2}+3d_{i2}$. Hence the coordinator obtains $G_1^F[3]=A_{12}+A_{22}, G_1^F[4]=\bar A_{12}+\bar A_{22}$, together with \(B_{12}+B_{22}\).

Similarly, the coordinator reads the first physical coordinate of $X_5,X_6,L_3,X_7,X_8,L_4$. This gives $A_{31}+A_{41}, \bar A_{31}+\bar A_{41}$, and also gives \(B_{31}+B_{41}\). Since the second final codeword is written in the logical order \(\Pi_2\), the first two physical prefix symbols become the third and fourth logical coordinates of \(G_2^F\), namely $G^F_{2}[3]=A_{31}+A_{41}, G^F_{2}[4]=\bar A_{31}+\bar A_{41}$.

The remaining prefix coordinates are obtained from the first initial global parity node: $G^F_{1}[1]=G^I_{1}[1]-G^F_{2}[3], G^F_{1}[2]=G^I_{1}[2]-G^F_{2}[4]$, and $G^F_{2}[1]=G^I_{1}[3]-G^F_{1}[3], G^F_{2}[2]=G^I_{1}[4]-G^F_{1}[4]$. Finally, the two tail coordinates are recovered from the first and third coordinates of \(G_2^I\), which contain the two tail piggybacks: $G^F_{1}[5] =G^I_{2}[1]-(B_{31}+B_{41}) =R_1+B_{11}+B_{21}$, and $G^F_{2}[5] =G^I_{2}[3]-(B_{12}+B_{22}) =R_2+B_{32}+B_{42}$. Thus both final global parity nodes are recovered.

A direct rank computation in Magma over \(\mathbb F_5\) gives $d(C^F)=3, d(C^I)=4$. Hence both the initial and final codes are optimal-distance LRCs.

The read bandwidth is $\gamma_R=12+6=18$. Here, \(12\) symbols are read from the unchanged information and local parity nodes, while \(6\) symbols are read from the initial global parity nodes. For these parameters, the lower bound in Theorem~\ref{thm:final_thm} is also \(18\), hence this conversion attains the optimal read bandwidth.

\subsubsection{General Construction}
Throughout this subsection, we assume \(g^F<g^I\le r\) and set  $\alpha_0=g^F+\ell, \alpha=\lambda(g^F+\ell)+(g^I-g^F)$. The construction requires an MDS array block with a prescribed aligned-recovery property: while optimally repairing the first \(g^F\) global parity nodes, the downloaded symbols must simultaneously determine the first \(g^F\) subsymbols of each of the remaining \(g^I-g^F\) global parity nodes. We first construct a small-subpacketization MDS array code having this property.

\begin{lemma}
\label{lem:fixed_parity_repair_aligned}
    Let \(r,\ell,g^F,g^I\) be positive integers with \(g^F<g^I\le r\), and set \(\alpha_0:=g^F+\ell\). Let \(\mathbb F_q\) be a finite field over which there exists a \((g^I+\ell)\times r\) superregular matrix. Such a matrix can be obtained, for example, from a Cauchy matrix whenever \(q\ge r+g^I+\ell\). Then there exists a systematic linear \([r+\ell+g^I,r,\alpha_0]_q\) MDS array code with nodes
    \[
    X_1,\ldots,X_r,\quad L_1,\ldots,L_\ell,\quad P_1,\ldots,P_{g^I}
    \]
    such that \(P_1,\ldots,P_{g^F}\) can be optimally repaired from \(X_1,\ldots,X_r,L_1,\ldots,L_\ell\) by downloading the first \(g^F\) subsymbols from each helper node. Moreover, the same downloads determine the first \(g^F\) subsymbols of each of \(P_{g^F+1},\ldots,P_{g^I}\).
\end{lemma}

\begin{proof}
    Let \(\Gamma=(\gamma_{\rho,i})_{\rho\in[g^I+\ell],\,i\in[r]}\) be a \((g^I+\ell)\times r\) superregular matrix over \(\mathbb F_q\). For each \(i\in[r]\), let \(X_i=(x_{i,1},\ldots,x_{i,\alpha_0})^\top \in\mathbb F_q^{\alpha_0}\). For each \(\rho\in[g^I+\ell]\), define \(Q_\rho=\sum_{i=1}^r\gamma_{\rho,i}X_i =(q_{\rho,1},\ldots,q_{\rho,\alpha_0})^\top \in\mathbb F_q^{\alpha_0}\) where \(q_{\rho,c}=\sum_{i=1}^r\gamma_{\rho,i}x_{i,c}\) for \(c\in[\alpha_0]\).

    Let \(e_1,\ldots,e_{\alpha_0}\) denote the standard basis vectors of \(\mathbb F_q^{\alpha_0}\). Define \(L_b:=Q_b+\sum_{h=1}^{g^F}q_{\ell+h,g^F+b}e_h\), \(b\in[\ell]\) and \(P_h:=Q_{\ell+h}\), \(h\in[g^I]\).

    We first prove that the resulting array code is MDS. Without the piggyback terms, each coordinate induces an \([r+g^I+\ell,r]\) scalar MDS code, since \(\Gamma\) is superregular. The piggyback terms only modify the first \(g^F\) coordinates by adding symbols from the last \(\ell\) coordinates. Hence, given any \(r\) surviving nodes, we first decode the last \(\ell\) coordinates \(g^F+1,\ldots,g^F+\ell\). These coordinates are unaffected by the piggyback terms, and therefore the corresponding scalar MDS codes can be decoded directly. In particular, all symbols $q_{\ell+h,g^F+b}, h\in[g^F],\ b\in[\ell]$ are recovered.

    We then decode the first \(g^F\) coordinates. For each \(h\in[g^F]\), the piggyback terms in coordinate \(h\) are now known and can be subtracted from the surviving local parity subsymbols. After removing these terms, the remaining subsymbols are exactly those of the unpiggybacked \(h\)-th scalar MDS code. Hence the first \(g^F\) coordinates can also be decoded. Therefore, any \(r\) surviving nodes determine the whole message, and the resulting array code is MDS.

    It remains to verify the repair and aligned-recovery property. Let \(\pi:\mathbb F_q^{\alpha_0}\to\mathbb F_q^{g^F}\) denote the projection onto the first \(g^F\) coordinates. To repair \(P_1,\ldots,P_{g^F}\), the decoder downloads \(\pi(X_i)\) from each information node \(X_i\) and \(\pi(L_b)\) from each local parity node \(L_b\). From the downloaded information-node subsymbols, the decoder can compute \(\pi(Q_\rho)=(q_{\rho,1},\ldots,q_{\rho,g^F})^\top\) for every \(\rho\in[g^I+\ell]\). Hence the first \(g^F\) coordinates of \(P_1,\ldots,P_{g^I}\) are recovered. Moreover, for each \(b\in[\ell]\),
    \[
    \pi(L_b)-\pi(Q_b) =  \bigl( q_{\ell+1,g^F+b},\ldots,q_{\ell+g^F,g^F+b} \bigr)^\top.
    \]
    Therefore, the \((g^F+b)\)-th coordinates of \(P_1,\ldots,P_{g^F}\) are recovered. Repeating this for every \(b\in[\ell]\) recovers all remaining coordinates of these global parity nodes. Thus \(P_1,\ldots,P_{g^F}\) are completely recovered, while the same downloads also determine the first \(g^F\) coordinates of \(P_{g^F+1},\ldots,P_{g^I}\).

    The total download is \(g^F(r+\ell)\), which equals the centralized repair cut-set bound
    \[
    \frac{g^F(r+\ell)}  {g^F+(r+\ell)-r}\alpha_0 = g^F(r+\ell),
    \]
    where \(\alpha_0=g^F+\ell\). Hence the repair of \(P_1,\ldots,P_{g^F}\) is optimal.
\end{proof}

We now construct the LRCC for the case \(g^F<g^I\). Since we take \(\alpha=\lambda\alpha_0+m =\lambda(g^F+\ell)+(g^I-g^F)\), each node in both the initial and final codes is a column vector in \(\mathbb F_q^\alpha\). We partition its coordinates into \(\lambda\) prefix blocks \(\mathcal B_1,\ldots,\mathcal B_\lambda\), together with one tail block \(\mathcal T\), where each prefix block \(\mathcal B_s\) has length \(\alpha_0=g^F+\ell\), while the tail block \(\mathcal T\) has length \(m=g^I-g^F\).

For each \(t\in[\lambda]\), the local groups of the \(t\)-th final codeword are indexed by \([\mu]^t\). For \(\tau\in[\mu]^t\), let $\iota(\tau):=\tau-(t-1)\mu\in[\mu]$, which denotes the local index of the group \(\tau\) inside the \(t\)-th final codeword.

For each \(\tau\in[\lambda\mu]\) and \(s\in[\lambda]\), take an independent copy of the MDS array code in Lemma~\ref{lem:fixed_parity_repair_aligned}. Use the \(r\) information nodes of this copy to define the \(\mathcal B_s\)-parts of \(X_{[r]^\tau}\), and use its \(\ell\) local parity nodes to define the \(\mathcal B_s\)-parts of \(L_{[\ell]^\tau}\). Denote the \(g^I\) global parity nodes produced by this copy, which will be used only as virtual nodes, by $P_{\tau,1}^{(s)},\ldots,P_{\tau,g^I}^{(s)} \in\mathbb F_q^{\alpha_0}$. Define
$$
\mathbf P_\tau^{(s)} := \big( (P_{\tau,1}^{(s)})^\top,\ldots, (P_{\tau,g^I}^{(s)})^\top \big)^\top \in\mathbb F_q^{g^I\alpha_0}.
$$

We next decompose the coordinate space \(\mathbb F_q^{g^I\alpha_0}\). For \(v\in\mathbb F_q^{g^I\alpha_0}\), write $v=\bigl(v_1^\top,\ldots,v_{g^I}^\top\bigr)^\top$, where $v_j=(v_{j,1},\ldots,v_{j,\alpha_0})^\top\in\mathbb F_q^{\alpha_0}$. Recall that \(\alpha_0=g^F+\ell\).  We decompose $\mathbb F_q^{g^I\alpha_0} = \mathcal F\oplus\mathcal E\oplus\mathcal W$ as follows. The subspace \(\mathcal F\) is supported on the first \(g^F\) \(\alpha_0\)-blocks, namely on the coordinates $v_{j,c}$ with \(j\in[g^F]\) and \(c\in[\alpha_0]\). The remaining \(m=g^I-g^F\) blocks are divided into two coordinate subspaces. The subspace \(\mathcal E\) is supported on $v_{g^F+b,a}, b\in[m],\ a\in[g^F]$ while \(\mathcal W\) is supported on  $v_{g^F+b,g^F+c}, b\in[m],\ c\in[\ell]$. Hence
$\dim\mathcal F=g^F(g^F+\ell), \dim\mathcal E=(g^I-g^F)g^F$ and $ \dim\mathcal W=(g^I-g^F)\ell$. Let $\mathcal U:=\mathcal F\oplus\mathcal E$, and denote the coordinate projections onto \(\mathcal F\), \(\mathcal E\), and \(\mathcal U\) by \(\pi_{\mathcal F}\), \(\pi_{\mathcal E}\), and \(\pi_{\mathcal U}\), respectively.

We view \(\pi_{\mathcal F}(v)\) as \(g^F\) consecutive blocks of length \(\alpha_0\). For \(\pi_{\mathcal E}(v)\), fix the block order
$$
\pi_{\mathcal E}(v) = \bigl( v_{g^F+1,1},\ldots,v_{g^F+1,g^F}; v_{g^F+2,1},\ldots,v_{g^F+2,g^F}; \ldots; v_{g^I,1},\ldots,v_{g^I,g^F} \bigr)^\top \in\mathbb F_q^{g^F m}.
$$
Accordingly, \(\pi_{\mathcal E}(v)\) is viewed as \(m\) consecutive blocks of length \(g^F\).

Let $\operatorname{emb}_{\mathcal E}:\mathbb F_q^{g^F m} \to\mathbb F_q^{g^I\alpha_0}$ denote the coordinate embedding onto \(\mathcal E\), with zeros on \(\mathcal F\) and \(\mathcal W\). Explicitly, for $z=\bigl( z_{1,1},\ldots,z_{1,g^F}; z_{2,1},\ldots,z_{2,g^F}; \ldots; z_{m,1},\ldots,z_{m,g^F} \bigr)^\top$, the vector \(\operatorname{emb}_{\mathcal E}(z)\) places \(z_{b,a}\) in coordinate \((g^F+b,a)\), namely the \(a\)-th coordinate of the \((g^F+b)\)-th \(\alpha_0\)-block, for \(b\in[m]\) and \(a\in[g^F]\).

We apply Lemma~\ref{lem:mixing_auxiliary_lrc} in Appendix~\ref{appe} to the MDS array code $\mathcal{C}$ constructed in Lemma~\ref{lem:fixed_parity_repair_aligned}, with the parameter \(g\) in the appendix specialized to \(g^I\). This yields block-scalar diagonal matrices \(A_i\), \(i\in[\mu]\), and \(B_t\), \(t\in[\lambda]\), which will be used in the construction below. We will later verify that, with these choices, both the initial and final codes are optimal-distance LRCs. 

In addition, the conversion procedure below requires every such matrix \(M\) to satisfy \(M\mathcal W\subseteq\mathcal W\). This condition holds automatically, since a block-scalar diagonal matrix scales each \(\alpha_0\)-block by a nonzero scalar and therefore preserves \(\mathcal W\). Consequently, \(\pi_{\mathcal U}(Mv)\) is determined by \(\pi_{\mathcal U}(v)\), and the induced map  \(x\mapsto\pi_{\mathcal U}(Mx)\) is invertible on \(\mathcal U\).

For \(t,s\in[\lambda]\), define
\[
\mathbf{G}_{t,s}:= \sum_{\tau\in[\mu]^t} A_{\iota(\tau)}\mathbf{P}_\tau^{(s)} \in\mathbb F_q^{g^I\alpha_0}.
\]
We first define the prefix parts of the final global parity nodes. For the \(t\)-th final codeword and the prefix block \(\mathcal B_s\), let
\[
\mathbf{G}^F_{[g^F]^t,s} := \pi_{\mathcal F}(\mathbf{G}_{t,s}) \in\mathbb F_q^{g^F\alpha_0}.
\]
Write
\[
\mathbf{G}^F_{[g^F]^t,s} = \bigl( (G^F_{(t-1)g^F+1,s})^\top,\ldots, (G^F_{tg^F,s})^\top \bigr)^\top, \qquad G^F_{i,s}\in\mathbb F_q^{\alpha_0},
\]
for each \(i\in[g^F]^t\), define the \(\mathcal B_s\)-part of the final global parity node \(G_i^F\) to be \(G^F_{i,s}\).

It remains to define the tail parts of the final global parity nodes. We first prepare scalar tail codes shared by the initial and final codewords. Choose a systematic scalar \([r+\ell+g^I,r]_q\) MDS code. For each \(\tau\in[\lambda\mu]\) and \(b\in[m]\), take an independent copy of this code. Use its \(r\) information symbols to define the \(b\)-th coordinates of the \(\mathcal T\)-parts of \(X_{[r]^\tau}\), and use its \(\ell\) local parity symbols to define the \(b\)-th coordinates of the \(\mathcal T\)-parts of \(L_{[\ell]^\tau}\).

For each \(b\in[m]\), apply Lemma~\ref{lem:mixing_auxiliary_lrc} in Appendix~\ref{appe} with \(g=g^I\) and \(\alpha=1\) to the \(\lambda\mu\) scalar MDS copies associated with the \(b\)-th tail coordinate. The copies can be mixed so that, for every \(t\in[\lambda]\), the \(\mu\) copies indexed by \([\mu]^t\), together with their global parity symbols, form a codeword of the same optimal-distance \((k^F,g^I,r,\ell)\)-LRC, while all \(\lambda\mu\) copies together form a codeword of an optimal-distance \((k^I,g^I,r,\ell)\)-LRC.

Puncturing the last \(m=g^I-g^F\) global parity symbols from the common final scalar code yields an optimal-distance \((k^F,g^F,r,\ell)\)-LRC. We use this punctured code to define the final tail parity symbols, while the scalar initial code will be used later to define the tail parts of the initial global parity nodes.

For \(t\in[\lambda]\), \(a\in[g^F]\), and \(b\in[m]\), let \(R_{t,a,b}^F\) denote the \(a\)-th global parity symbol of the \(t\)-th final scalar codeword associated with the \(b\)-th tail coordinate. For \(i=(t-1)g^F+a\in[g^F]^t\), define the tail parity vector
\[
 R_i^F := \bigl( R_{t,a,1}^F,\ldots,R_{t,a,m}^F \bigr)^\top \in\mathbb F_q^m.
\]
We also arrange these symbols in the \(\mathcal E\)-coordinate order as
\[
\mathbf R_{t,\mathcal E}^F := \bigl( R_{t,1,1}^F,\ldots,R_{t,g^F,1}^F; \ R_{t,1,2}^F,\ldots,R_{t,g^F,2}^F; \ \ldots; \ R_{t,1,m}^F,\ldots,R_{t,g^F,m}^F \bigr)^\top \in\mathbb F_q^{g^F m}.
\]

For each \(t\in[\lambda]\), define the symbols \(E_{t,b,a}^F\), \(b\in[m]\), \(a\in[g^F]\), by writing the \(\mathcal E\)-part of the diagonal mixed vector as
\[
\pi_{\mathcal E}(\mathbf G_{t,t}) = \bigl( E_{t,1,1}^F,\ldots,E_{t,1,g^F}^F; \ E_{t,2,1}^F,\ldots,E_{t,2,g^F}^F; \ \ldots; \ E_{t,m,1}^F,\ldots,E_{t,m,g^F}^F \bigr)^\top.
\]
For \(i=(t-1)g^F+a\in[g^F]^t\), define the tail piggyback vector
\[
 E_i^F := \bigl( E_{t,1,a}^F,\ldots,E_{t,m,a}^F \bigr)^\top \in\mathbb F_q^m.
\]

Under the fixed \(\mathcal E\)-coordinate order, the \((b,a)\)-th coordinate of $\pi_{\mathcal E}(\mathbf G_{t,t}) + \mathbf R_{t,\mathcal E}^F$ is $E_{t,b,a}^F+R_{t,a,b}^F.$
Hence, after regrouping the coordinates by the global parity index \(a\), define the \(\mathcal T\)-part of \(G_i^F\) to be $R_i^F+E_i^F.$

Combining the prefix and tail parts, for each \(i\in[g^F]^t\), define
\[
G_i^F := \bigl( (G^F_{i,1})^\top,\ldots, (G^F_{i,\lambda})^\top, (R_i^F+E_i^F)^\top \bigr)^\top \in\mathbb F_q^\alpha.
\]
The nodes \(G_i^F\), \(i\in[g^F]^t\), are the \(g^F\) final global parity nodes of the \(t\)-th final codeword.

We next verify that all final codewords are codewords of the same final LRC. The preceding definition uses the physical block order $\mathcal B_1,\ldots,\mathcal B_\lambda,\mathcal T.$
For each \(t\in[\lambda]\), define the cyclic permutation \(\sigma_t:[\lambda]\to[\lambda]\) by $\sigma_t(j) := 1+\bigl((t+j-2)\bmod\lambda\bigr),
 j\in[\lambda]$.
We describe the \(t\)-th final codeword in the logical block order $\mathcal B_{\sigma_t(1)},\ldots, \mathcal B_{\sigma_t(\lambda)},\mathcal T$.
Since \(\sigma_t(1)=t\), the diagonal mixed vector \(\mathbf G_{t,t}\), whose \(\mathcal E\)-part defines the tail piggyback, corresponds to the first logical prefix block \(\mathcal B_t\).

This cyclic ordering only relabels the subsymbol coordinates: the physical storage order remains unchanged, and no preserved information or local parity node is moved or rewritten. In logical coordinates, every final codeword uses the same prefix encoding rule determined by \(\mathcal C\) and \(A_1,\ldots,A_\mu\), the same collection of scalar tail codes, and the same rule for deriving the tail piggyback from the \(\mathcal E\)-part associated with the first logical prefix block. Hence all final codewords are codewords of the same final LRC.

We next define the initial global parity nodes. For each \(s\in[\lambda]\), define the stacked vector of their \(\mathcal B_s\)-parts by
\[
\mathbf G^I_{[g^I],s} := \sum_{t\in[\lambda]} B_t\mathbf G_{t,s} + B_s\operatorname{emb}_{\mathcal E} \bigl(\mathbf R_{s,\mathcal E}^F\bigr) \in\mathbb F_q^{g^I\alpha_0}.
\]
Write
\[
\mathbf G^I_{[g^I],s} = \bigl( (G^I_{1,s})^\top,\ldots, (G^I_{g^I,s})^\top \bigr)^\top, \qquad G^I_{a,s}\in\mathbb F_q^{\alpha_0},
\]
for each \(a\in[g^I]\), define the \(\mathcal B_s\)-part of the initial global parity node \(G_a^I\) to be \(G_{a,s}^I\).

For the tail block \(\mathcal T\), use the initial scalar codeword associated with each \(b\in[m]\). For \(a\in[g^I]\) and \(b\in[m]\), let \(R_{a,b}^I\) denote the \(a\)-th global parity symbol in the initial scalar codeword associated with the \(b\)-th tail coordinate. Define
\[
 R^I_a := (R^I_{a,1},\ldots,R^I_{a,m})^\top \in\mathbb F_q^m.
\]
Combining the prefix and tail parts, for each \(a\in[g^I]\), define
\[
G^I_a := \bigl( (G^I_{a,1})^\top,\ldots, (G^I_{a,\lambda})^\top, ( R^I_a)^\top \bigr)^\top \in\mathbb F_q^\alpha.
\]
The nodes \(G_1^I,\ldots,G_{g^I}^I\) are the initial global parity nodes. By construction, all information and local parity nodes are preserved during conversion. Hence the conversion is stable.

We now verify that the initial and final codes are optimal-distance LRCs. Before adding the prefix-tail piggybacks, each prefix block of every final codeword is obtained by puncturing the last \(m=g^I-g^F\) global parity nodes from \(\mathcal L_{\mathcal C}(A_1,\ldots,A_\mu)\). Since the latter is an optimal-distance \((k^F,g^I,r,\ell,\alpha_0)\)-LRC, the punctured code has minimum distance at least $g^I+\ell+1-m=g^F+\ell+1$.
The Singleton-type bound gives the reverse inequality, so each prefix block is an optimal-distance \((k^F,g^F,r,\ell,\alpha_0)\)-LRC. Each tail coordinate is encoded by the same optimal-distance scalar \((k^F,g^F,r,\ell)\)-LRC. Hence every erasure pattern of at most \(g^F+\ell\) nodes can be decoded independently in all prefix blocks and tail coordinates, and their concatenation is an optimal-distance \((k^F,g^F,r,\ell,\alpha)\)-LRC. The logical cyclic ordering does not affect this property.

In the actual final code, \(E_i^F\) is piggybacked from the diagonal prefix block onto the \(\mathcal T\)-part of \(G_i^F\). The prefix blocks are decoded first, after which the vectors \(E_i^F\) are computed and subtracted from the corresponding tail parts. The scalar tail codes can then be decoded. Thus the piggybacks do not change the optimal-distance property of the final code.

Similarly, before adding the prefix-tail piggybacks, each prefix block of the initial code is the optimal-distance LRC $\mathcal L_{\mathcal C} \bigl( B_tA_i:t\in[\lambda],\,i\in[\mu] \bigr)$, while each tail coordinate is encoded by the optimal-distance scalar initial LRC. By the same coordinatewise decoding argument, their concatenation is an optimal-distance \((k^I,g^I,r,\ell,\alpha)\)-LRC. In the actual initial code, the term $B_s\operatorname{emb}_{\mathcal E} \bigl(\mathbf R_{s,\mathcal E}^F\bigr)$ is added to the \(\mathcal B_s\)-part. One first decodes the scalar initial tail codes, which determine \(\mathbf R_{s,\mathcal E}^F\), and then subtracts this piggyback term before decoding the corresponding prefix block. Hence the initial code is also optimal-distance.

We now describe the conversion procedure. First, for each \(s\in[\lambda]\), the coordinator downloads \(\pi_{\mathcal U}(\mathbf G^I_{[g^I],s})\) from the \(\mathcal B_s\)-parts of the initial global parity nodes. Since $\dim\mathcal U = g^F(g^F+\ell)+(g^I-g^F)g^F = g^F(g^I+\ell),$ this step reads \(\lambda g^F(g^I+\ell)\) symbols from the initial global parity nodes.

Next, for every \(t,s\in[\lambda]\) with \(t\ne s\), and for every
local group \(\tau\in[\mu]^t\), the coordinator downloads the first \(g^F\) subsymbols of the \(\mathcal B_s\)-part of every node in \(X_{[r]^\tau}\) and \(L_{[\ell]^\tau}\). By Lemma~\ref{lem:fixed_parity_repair_aligned}, these downloads determine \(\pi_{\mathcal U}\bigl(\mathbf P_\tau^{(s)}\bigr)\). Since \(A_{\iota(\tau)}\) and \(B_tA_{\iota(\tau)}\) preserve \(\mathcal W\), the coordinator can compute $\pi_{\mathcal U} \bigl(A_{\iota(\tau)}\mathbf P_\tau^{(s)}\bigr) ~\text{and}~ \pi_{\mathcal U} \bigl(B_tA_{\iota(\tau)}\mathbf P_\tau^{(s)}\bigr).$ Summing over \(\tau\in[\mu]^t\), the coordinator obtains \(\pi_{\mathcal U}(\mathbf G_{t,s})\) and \(\pi_{\mathcal U}(B_t\mathbf G_{t,s})\) for every \(t\ne s\). In particular, it recovers the off-diagonal final prefix part $\mathbf G^F_{[g^F]^t,s}=\pi_{\mathcal F}(\mathbf G_{t,s})$ while \(\pi_{\mathcal U}(B_t\mathbf G_{t,s})\) is retained to remove the corresponding off-diagonal term from \(\mathbf G^I_{[g^I],s}\). This step reads $\lambda(\lambda-1)g^F(k^F+\mu\ell)$ symbols from the preserved information and local parity nodes.

It remains to recover the diagonal prefix parts and the tail parts. Fix \(s\in[\lambda]\). From the downloaded \(\pi_{\mathcal U}(\mathbf G^I_{[g^I],s})\), the coordinator subtracts the known off-diagonal terms \(\pi_{\mathcal U}(B_t\mathbf G_{t,s})\), \(t\ne s\), and obtains $\pi_{\mathcal U}\!\left(B_s\left( \mathbf G_{s,s} + \operatorname{emb}_{\mathcal E} \bigl(\mathbf R_{s,\mathcal E}^F\bigr) \right) \right).$ Since the map induced by \(B_s\) on \(\mathcal U\) is invertible, the coordinator recovers $\pi_{\mathcal U}\!\left( \mathbf G_{s,s} + \operatorname{emb}_{\mathcal E} \bigl(\mathbf R_{s,\mathcal E}^F\bigr) \right).$  Its \(\mathcal F\)-part is $\pi_{\mathcal F}(\mathbf G_{s,s}) = \mathbf G^F_{[g^F]^s,s},$ which gives the diagonal prefix part of the \(s\)-th final codeword. Its \(\mathcal E\)-part is $\pi_{\mathcal E}(\mathbf G_{s,s}) + \mathbf R_{s,\mathcal E}^F.$
In the fixed \(\mathcal E\)-coordinate order, the \((b,a)\)-th coordinate of this vector is $E_{s,b,a}^F+R_{s,a,b}^F, b\in[m], a\in[g^F]$. Therefore, for $i=(s-1)g^F+a\in[g^F]^s$, regrouping these coordinates by the global parity index \(a\) gives $\bigl( R_{s,a,1}^F+E_{s,1,a}^F, \ldots, R_{s,a,m}^F+E_{s,m,a}^F \bigr)^\top = R_i^F+E_i^F$.
Together with the prefix parts \(\mathbf G^F_{[g^F]^s,j}\), \(j\in[\lambda]\), these tail vectors determine the final global parity nodes \(G_i^F\), \(i\in[g^F]^s\). Repeating this procedure for every \(s\in[\lambda]\) generates all final global parity nodes.

The total read bandwidth is $\gamma_R =\lambda(\lambda-1)g^F(k^F+\mu\ell) +\lambda g^F(g^I+\ell) =\lambda g^F\big((\lambda-1)(k^F+\mu\ell)+g^I+\ell\big)$. Since $\alpha =\lambda(g^F+\ell)+(g^I-g^F) =(\lambda-1)(g^F+\ell)+g^I+\ell$, we obtain
$$
\gamma_R = \lambda g^F \frac{ (\lambda-1)(k^F+\mu\ell)+g^I+\ell }{ (\lambda-1)(g^F+\ell)+g^I+\ell}\alpha.
$$
When \(g^F<g^I\), this is exactly the lower bound in Theorem~\ref{thm:final_thm}. Hence the construction attains the optimal read bandwidth.

\subsection{Overall Achievability}

The three constructions above yield the following overall achievability result.

\begin{theorem}
\label{thm:overall_achievability}
    Let \(r,\ell,g^I,g^F,k^F,\lambda\) be positive integers such that $r\mid k^F, \lambda\ge2, g^I,g^F\le r. $ Set $k^I=\lambda k^F, \mu=\frac{k^F}{r}. $ Over any sufficiently large finite field, there exists a stable optimal-distance $(k^I,g^I,r,\ell;\,k^F,g^F,r,\ell,\alpha)\text{-LRCC}$ in the global split regime whose read bandwidth is
    \[
    \gamma_R=\begin{cases} \displaystyle\lambda g^F\frac{(\lambda-1)(k^F+\mu\ell)+g^I+\ell}{(\lambda-1)(g^F+\ell)+g^I+\ell}\alpha,
        & \text{if }~g^F\le g^I,\\[4mm]\displaystyle\lambda\frac{g^F}{g^F+\ell}(k^F+\mu\ell)\alpha-g^I\left(\frac{k^F+\mu\ell}{g^F+\ell}-1\right)\alpha,& \text{if }~ g^F>g^I.\end{cases}
    \] Moreover, the subpacketization can be chosen as
    \[
    \alpha
    =
    \begin{cases}
        g^F+\ell,
        & g^F\ge g^I,\\[1mm]
        \lambda(g^F+\ell)+(g^I-g^F),
        & g^F<g^I.
    \end{cases}
    \]
\end{theorem}

\begin{proof}
The cases \(g^F=g^I\), \(g^F>g^I\), and \(g^F<g^I\) are
established by the constructions in Sections~\ref{subsec5.1}, \ref{subsec5.2}, and \ref{subsec5.3},
respectively.
\end{proof}

Combining Theorem~\ref{thm:final_thm} and
Theorem~\ref{thm:overall_achievability} gives the exact optimal read
bandwidth for stable optimal-distance LRCCs throughout the
parameter range \(g^I,g^F\le r\) over sufficiently large finite
fields.

\section{Conclusion}\label{sec6}

In this work, we studied the fundamental limits on the read bandwidth of stable optimal-distance LRCCs in the global split regime over the parameter range \(g^I,g^F\le r\). Using an information-theoretic approach, we derived, to the best of our knowledge, the first read-bandwidth lower bounds for this setting, without imposing any linearity assumption on the initial code, the final code, or the conversion procedure. We then developed MDS array codes with prescribed repair or alignment properties using piggybacking techniques and used them as building blocks for stable optimal-distance LRCCs that attain the corresponding lower bounds over sufficiently large finite fields. Together, the converse and achievability results completely characterize the optimal read bandwidth throughout the parameter range \(g^I,g^F\le r\).

A natural problem is to study whether lower bounds and matching constructions can be obtained in this framework when \(g^I\) or \(g^F\) is larger than \(r\). Another direction is to consider conversions between different types of codes, such as conversions between LRCs with different locality parameters, MDS-to-LRC conversions, and LRC-to-MDS conversions, and to investigate whether information-theoretic methods can be used to derive lower bounds together with corresponding optimal constructions. It would also be interesting to study conversion problems for broader classes of codes, including Reed--Muller codes, regenerating codes, algebraic geometry codes, and other related code families.

\appendices
\section{Proof of the Mixing Lemma}
\label{appe}

In this appendix, we give the formal definition of the auxiliary mixing construction used in Section~\ref{sec5} and prove the required matrix-existence result. Let \(r,\ell,g,\alpha\) be positive integers, and let \(\mathcal C\) be a systematic linear \([r+\ell+g,r,\alpha]_q\) MDS array code whose nodes are ordered as $X_1,\ldots,X_r,\ L_1,\ldots,L_\ell,\ P_1,\ldots,P_g $.
For a message vector \(x\in(\mathbb F_q^\alpha)^r\), denote the corresponding node contents by \(X_i(x)\), \(L_b(x)\), and \(P_a(x)\), and define
\[
    \mathbf P(x) := \bigl(P_1(x)^\top,\ldots,P_g(x)^\top\bigr)^\top \in\mathbb F_q^{g\alpha}.
\]
When the underlying message is clear, we suppress the argument \(x\) and simply write \(X_i\), \(L_b\), \(P_a\), and \(\mathbf P\).

For \(c_1,\ldots,c_g\in\mathbb F_q\), define the block-scalar diagonal matrix
\[
    D(c_1,\ldots,c_g):= \operatorname{diag}(c_1I_\alpha,\ldots,c_gI_\alpha) \in \mathbb F_q^{g\alpha\times g\alpha}.
\]
Let
\[
    \mathcal D_g := \{D(c_1,\ldots,c_g): c_1,\ldots,c_g\in\mathbb F_q^*\}\subseteq\operatorname{GL}_{g\alpha}(\mathbb F_q).
\]

For \(N\ge 1\) and \(M_1,\ldots,M_N\in\mathcal D_g\), define the code \(\mathcal L_{\mathcal{C}}(M_1,\ldots,M_N)\) as follows. Take \(N\) copies \(\mathcal C^{(1)},\ldots,\mathcal C^{(N)}\) of \(\mathcal C\), each encoding a disjoint collection of \(r\) message blocks. For \(i\in[N]\), denote the nodes of the \(i\)-th copy by
\[
    X^{(i)}_1,\ldots,X^{(i)}_r,\quad L^{(i)}_1,\ldots,L^{(i)}_\ell,\quad P^{(i)}_1,\ldots,P^{(i)}_g .
\]
The nodes \(X^{(i)}_1,\ldots,X^{(i)}_r,L^{(i)}_1,\ldots,L^{(i)}_\ell\) form the \(i\)-th local group of \(\mathcal L_{\mathcal C}(M_1,\ldots,M_N)\). The nodes \(P^{(i)}_1,\ldots,P^{(i)}_g\) are virtual nodes: they are not stored and are used only to define the global parity nodes. Write
\[
    \mathbf P^{(i)} := \bigl((P^{(i)}_1)^\top,\ldots,(P^{(i)}_g)^\top\bigr)^\top \in\mathbb F_q^{g\alpha}.
\]
Define
\[
    \mathbf G := \sum_{i=1}^N M_i\mathbf P^{(i)} \in\mathbb F_q^{g\alpha},
\]
and write
\[
    \mathbf G=(G_1^\top,\ldots,G_g^\top)^\top, \qquad G_a\in\mathbb F_q^\alpha .
\]
The nodes \(G_1,\ldots,G_g\) are stored as the global parity nodes of \(\mathcal L_{\mathcal{C}}(M_1,\ldots,M_N)\). 

Since the puncturing of \(\mathcal C\) onto $X_1,\ldots,X_r,L_1,\ldots,L_\ell$ is an \([r+\ell,r,\alpha]_q\) MDS array code, every local group above has minimum distance \(\ell+1\). Therefore, \(\mathcal L_{\mathcal C}(M_1,\ldots,M_N)\), consisting of these \(N\) local groups together with the global parity nodes \(G_1,\ldots,G_g\), is a systematic \((Nr,g,r,\ell,\alpha)\)-LRC.

\begin{lemma}[Mixing Lemma]
\label{lem:mixing_auxiliary_lrc}
    Let \(\mu\) and \(\lambda\) be positive integers. If \(q\) is sufficiently large, then there exist matrices $A_1,\ldots,A_\mu$, $B_1,\ldots,B_\lambda\in\mathcal D_g$ such that $\mathcal L_{\mathcal C}(A_1,\ldots,A_\mu)$ is an optimal-distance \((\mu r,g,r,\ell,\alpha)\)-LRC, and $\mathcal L_{\mathcal C} \bigl(B_tA_i:\ t\in[\lambda],\ i\in[\mu]\bigr)$ is an optimal-distance \((\lambda \mu r,g,r,\ell,\alpha)\)-LRC.
\end{lemma}

\begin{proof}
    For \(\mathcal{Q}\subseteq[g]\), let $\pi_\mathcal{Q}:\mathbb F_q^{g\alpha}\longrightarrow \mathbb F_q^{|\mathcal{Q}|\alpha} $ denote the projection onto the \(\alpha\)-blocks indexed by \(\mathcal Q\). Set $Z_1,\ldots,Z_{r+\ell}:=X_1,\ldots,X_r,L_1,\ldots,L_\ell .$ For \(\mathcal{E}\subseteq[r+\ell]\), define $h(\mathcal{E}):=\max\{|\mathcal{E}|-\ell,0\},$ and let
    \[
    \mathcal K_\mathcal{E}:=\{x\in(\mathbb F_q^\alpha)^r:Z_j(x)=0\ \text{for all }j\in[r+\ell]\setminus \mathcal{E}\}
    \]
    and
    \[
    \mathcal R_\mathcal{E}:= \{\mathbf P(x):x\in\mathcal K_\mathcal{E}\} \subseteq \mathbb F_q^{g\alpha}.
    \]

    We first record a basic consequence of the MDS property.  Let \(h:=h(\mathcal E)\le g\). Then \(\dim \mathcal R_\mathcal{E}=h\alpha\), and for every \(\mathcal{Q}\subseteq[g]\) with \(|\mathcal{Q}|=h\), the restriction map $\pi_\mathcal{Q}:\mathcal R_\mathcal{E}\longrightarrow \mathbb F_q^{h\alpha}$ is an isomorphism. Indeed, suppose first that \(|\mathcal E|\le\ell\). At least \(r\) local nodes survive, and hence the MDS property gives $ \mathcal K_{\mathcal E}=\{0\}, \mathcal R_{\mathcal E}=\{0\}.$ The assertion therefore follows in this case. Now assume \(|\mathcal{E}|=\ell+h\) with \(h>0\). Exactly \(r-h\) local nodes survive. Extend these nodes to any set of \(r\) nodes of \(\mathcal C\). By the MDS property, the projection onto the enlarged set has rank \(r\alpha\). Since the additional \(h\) nodes contribute at most \(h\alpha\) dimensions, the projection onto the original \(r-h\) surviving local nodes has full rank \((r-h)\alpha\). Therefore, by the rank-nullity theorem, $\dim \mathcal K_\mathcal{E}=r\alpha-(r-h)\alpha=h\alpha. $ Fix \(\mathcal Q\subseteq[g]\) with \(|\mathcal Q|=h\). The \(r-h\) surviving local nodes together with the global parity nodes \(P_a\), \(a\in\mathcal Q\), form \(r\) nodes of \(\mathcal C\). Therefore, if \(x\in\mathcal K_{\mathcal E}\) satisfies $\pi_{\mathcal Q}\mathbf P(x)=0,$ then these \(r\) node values are all zero, and the MDS property implies \(x=0\). Thus $x\longmapsto\pi_{\mathcal Q}\mathbf P(x)$ is injective on \(\mathcal K_{\mathcal E}\). It follows that \(x\mapsto\mathbf P(x)\) is injective on \(\mathcal K_{\mathcal E}\), and hence $ \dim\mathcal R_{\mathcal E} = \dim\mathcal K_{\mathcal E} =
    h\alpha.$ Moreover, the restriction of \(\pi_{\mathcal Q}\) to \(\mathcal R_{\mathcal E}\) is injective. Since its domain and codomain both have dimension \(h\alpha\), it is an isomorphism.

    We next give a criterion tailored to the optimal-distance requirement. Let \(\Omega_N\) be the collection of all tuples $\omega = (\mathcal E_1,\ldots,\mathcal E_N$, $\mathcal S)$ where $\mathcal E_i\subseteq[r+\ell], \mathcal S\subseteq[g]$ with $ \sum_{i=1}^N|\mathcal E_i|+|\mathcal S|=g+\ell$. For \(\omega\in\Omega_N\), set $h_i:=h(\mathcal E_i), H:=\sum_{i=1}^N h_i$ and $ s:=|\mathcal S|$. Fix \(M_1,\ldots,M_N\in\mathcal D_g\). Suppose that, for every \(\omega\in\Omega_N\) with \(H>0\), the map $\Phi_\omega: \prod_{i:h_i>0}\mathcal R_{\mathcal E_i} \longrightarrow \mathbb F_q^{(g-s)\alpha}$ defined by $\Phi_\omega\bigl((u_i)_{i:h_i>0}\bigr) := \sum_{i:h_i>0} \pi_{[g]\setminus\mathcal S}M_i u_i$ is injective. Then \(\mathcal L_{\mathcal C}(M_1,\ldots,M_N)\) is an optimal-distance \((Nr,g,r,\ell,\alpha)\)-LRC.

    Indeed, it suffices to correct every erasure pattern of size exactly \(g+\ell\), since any smaller pattern can be enlarged to one of this size. Fix \(\omega=(\mathcal E_1,\ldots,\mathcal E_N,\mathcal S)\in\Omega_N\). If \(H=0\), then every local group has at most \(\ell\) erased nodes, so all local erasures can be recovered locally and the erased global parity nodes can then be recomputed. Now suppose \(H>0\). Since at least one \(h_i\) is positive,
    \[
    \begin{aligned}
    \sum_{i=1}^N|\mathcal E_i|&=\sum_{i:h_i>0}(\ell+h_i)+\sum_{i:h_i=0}|\mathcal E_i| \ge\ell+H.
    \end{aligned}
    \] Together with $\sum_{i=1}^N|\mathcal E_i|+s=g+\ell,$ this gives \(H\le g-s\). Suppose that two codewords agree on all surviving nodes. For the \(i\)-th local group, let \(x_i\) be the difference of the corresponding message vectors. Then $x_i\in\mathcal K_{\mathcal E_i}.$ If \(h_i=0\), then \(\mathcal K_{\mathcal E_i}=\{0\}\), and hence \(x_i=0\). For every \(i\) with \(h_i>0\), set $u_i:=\mathbf P(x_i)\in\mathcal R_{\mathcal E_i}.$ The surviving global parity nodes therefore give $\sum_{i:h_i>0} \pi_{[g]\setminus\mathcal S}M_i u_i =0.$ By the injectivity of \(\Phi_\omega\), \(u_i=0\) for every \(i\) with \(h_i>0\). Since the preceding MDS consequence shows that \(x\mapsto\mathbf P(x)\) is injective on \(\mathcal K_{\mathcal E_i}\), we also have \(x_i=0\) for every \(i\) with \(h_i>0\). Thus all \(x_i\) are zero, so the two codewords coincide. Therefore every erasure pattern of size at most \(g+\ell\) is correctable, and $d\ge g+\ell+1$. The Singleton-type bound gives the reverse inequality, and hence the code is optimal-distance.

   We now show that the injectivity conditions above can be enforced by block-scalar diagonal matrices. Fix $\omega = (\mathcal E_1,\ldots,\mathcal E_N,\mathcal S)$ $ \in\Omega_N $ with \(H>0\). Since \(H\le g-|\mathcal S|\), choose and fix a set $\mathcal Q_\omega \subseteq[g]\setminus\mathcal S,  |\mathcal Q_\omega|=H.$ For each \(i\) with \(h_i>0\), let \(U_i\) be a matrix whose columns form a basis of \(\mathcal R_{\mathcal E_i}\), and let $i_1<\cdots<i_v$ be all indices satisfying \(h_i>0\). For indeterminates \(m_{i,j}\), we use the same notation $ M_i=D(m_{i,1},\ldots,m_{i,g})$ for the corresponding block-diagonal matrix over the polynomial ring \(\mathbb F_q[\{m_{i,j}\}_{i\in[N],\,j\in[g]}]\). Define $\Delta_\omega(M_1,\ldots,M_N) := \det
    \begin{bmatrix}
        \pi_{\mathcal Q_\omega}M_{i_1}U_{i_1}&\cdots&\pi_{\mathcal Q_\omega}M_{i_v}U_{i_v}
    \end{bmatrix}.$ The concatenated matrix has \(H\alpha\) rows and $\sum_{j=1}^v h_{i_j}\alpha=H\alpha $ columns, and is therefore square. Hence \(\Delta_\omega\) is a polynomial in the variables \(m_{i,j}\). At any assignment for which \(\Delta_\omega\ne0\), the map $(u_i)_{i:h_i>0}  \longmapsto \sum_{i:h_i>0} \pi_{\mathcal Q_\omega}M_i u_i$ is injective. Since \(\mathcal Q_\omega\subseteq[g]\setminus\mathcal S\), this implies the injectivity of \(\Phi_\omega\).

    It remains to show that \(\Delta_\omega\) is not the zero polynomial. Choose a partition $ \mathcal Q_\omega = \bigsqcup_{i:h_i>0}\mathcal Q_i, |\mathcal Q_i|=h_i.$ By the preceding MDS consequence, $\det\bigl(\pi_{\mathcal Q_i}U_i\bigr)\ne0$ for every \(i\) with \(h_i>0\). For each such \(i\), specialize the variables indexed by \(\mathcal Q_\omega\) as $ m_{i,j} =
    \begin{cases}
        1, & j\in\mathcal Q_i,\\
        0, & j\in\mathcal Q_\omega\setminus\mathcal Q_i.
    \end{cases}$
    The variables with \(j\notin\mathcal Q_\omega\) may be assigned arbitrarily. After permuting the row blocks according to the above partition, the matrix defining \(\Delta_\omega\) becomes block diagonal with diagonal blocks $\pi_{\mathcal Q_i}U_i, \text{for}~ i ~\text{with}~h_i>0.$ Therefore, under this specialization, $\Delta_\omega =  \pm  \prod_{i:h_i>0}  \det\bigl(\pi_{\mathcal Q_i}U_i\bigr)  \ne0.$ Thus \(\Delta_\omega\) is not the zero polynomial. The specializations used here and below serve only to show that the determinant polynomials are nonzero; the specialized matrices need not belong to \(\mathcal D_g\). The factor \(\mathcal P_{\mathrm{inv}}\) introduced below will enforce the invertibility of the final matrices.

    For convenience, define \[ \Omega_N^+ :=\left\{   (\mathcal E_1,\ldots,\mathcal E_N,\mathcal S)\in\Omega_N: \sum_{i=1}^N h(\mathcal E_i)>0\right\}. \] We first apply the preceding construction to the final code. For indeterminates \(a_{i,j}\), write $A_i=D(a_{i,1},\ldots,a_{i,g}), i\in[\mu],$ and define $\mathcal P_F(A_1,\ldots,A_\mu) := \prod_{\omega\in\Omega_\mu^+} \Delta_\omega(A_1,\ldots,A_\mu).$ By the preceding argument, every factor in this finite product is a nonzero polynomial. Hence \(\mathcal P_F\) is nonzero.

    It remains to handle the product form \(B_tA_i\) for the initial code. For indeterminates \(b_{t,j}\), write $B_t=D(b_{t,1},\ldots,b_{t,g}), t\in[\lambda].$ Index the \(\lambda\mu\) local groups by \((t,i)\in[\lambda]\times[\mu]\), and let \(\Omega_{\mathrm I}\) be the collection of all tuples $\omega =\bigl( \{\mathcal E_{t,i}\}_{t\in[\lambda],\,i\in[\mu]}, \mathcal S \bigr)$ such that $\mathcal E_{t,i}\subseteq[r+\ell], \mathcal S\subseteq[g],$ and $\sum_{t=1}^{\lambda}\sum_{i=1}^{\mu}|\mathcal E_{t,i}|+|\mathcal S|=g+\ell.$ For \(\omega\in\Omega_{\mathrm I}\), set $h_{t,i}:=h(\mathcal E_{t,i}), H_\omega := \sum_{t=1}^{\lambda}\sum_{i=1}^{\mu}h_{t,i}.$ Whenever \(H_\omega>0\), the same counting argument as above gives $ H_\omega\le g-|\mathcal S|.$ Fix a set $\mathcal Q_\omega \subseteq[g]\setminus\mathcal S, |\mathcal Q_\omega|=H_\omega.$

    For each \((t,i)\) with \(h_{t,i}>0\), let \(U_{t,i}\) be a matrix whose columns form a basis of \(\mathcal R_{\mathcal E_{t,i}}\). List all such pairs in lexicographic order as $(t_1,i_1),\ldots,(t_v,i_v),$ and define $ \Delta_\omega^I (A_1,\ldots,A_\mu,B_1,\ldots,B_\lambda) := \det
    \begin{bmatrix}
        \pi_{\mathcal Q_\omega} B_{t_1}A_{i_1}U_{t_1,i_1}&\cdots
        &\pi_{\mathcal Q_\omega}B_{t_v}A_{i_v}U_{t_v,i_v}
    \end{bmatrix}.$ The concatenated matrix has \(H_\omega\alpha\) rows and $\sum_{j=1}^v h_{t_j,i_j}\alpha = H_\omega\alpha$ columns, and is therefore square. Hence \(\Delta_\omega^I\) is a polynomial in the variables \(a_{i,j}\) and \(b_{t,j}\). At any assignment for which \(\Delta_\omega^I\ne0\), the map
    \[
    (u_{t,i})_{h_{t,i}>0} \longmapsto\sum_{(t,i):\,h_{t,i}>0}\pi_{\mathcal Q_\omega}B_tA_i u_{t,i}
    \]
    is injective. Since \(\mathcal Q_\omega\subseteq[g]\setminus\mathcal S\), this implies the injectivity condition in the preceding criterion for the erasure pattern \(\omega\).

    Now we claim that \(\Delta_\omega^I\) is not the zero polynomial. Choose a partition $\mathcal Q_\omega = \bigsqcup_{(t,i):\,h_{t,i}>0} \mathcal Q_{t,i}, |\mathcal Q_{t,i}|=h_{t,i}.$ By the preceding MDS consequence, $\det\bigl(\pi_{\mathcal Q_{t,i}}U_{t,i}\bigr) \ne0$ for every \((t,i)\) with \(h_{t,i}>0\).

    For each \(j\in\mathcal Q_{t,i}\), specialize $a_{i,j}=b_{t,j}=1,$ and set $a_{i',j}=0 ~\text{for every }i'\ne i, b_{t',j}=0 ~\text{for every }t'\ne t.$ After permuting the row blocks according to the partition of \(\mathcal Q_\omega\), the matrix defining \(\Delta_\omega^I\) becomes block diagonal with diagonal blocks $\pi_{\mathcal Q_{t,i}}U_{t,i}, (t,i):h_{t,i}>0.$ Therefore, under this specialization, $\Delta_\omega^I = \pm \prod_{(t,i):\,h_{t,i}>0} \det\bigl( \pi_{\mathcal Q_{t,i}}U_{t,i} \bigr) \ne0.$ Thus \(\Delta_\omega^I\) is not the zero polynomial.

    Let $\Omega_{\mathrm I}^+:= \left\{ \omega\in\Omega_{\mathrm I}: H_\omega>0 \right\}.$ Define $\mathcal P_I (A_1,\ldots,A_\mu,B_1,\ldots,B_\lambda) := \prod_{\omega\in\Omega_{\mathrm I}^+} \Delta_\omega^I (A_1,\ldots,A_\mu,B_1,\ldots,B_\lambda),$ where an empty product is understood to be \(1\). Since every factor is a nonzero polynomial and the polynomial ring over \(\mathbb F_q\) is an integral domain, \(\mathcal P_I\) is nonzero.

    Finally, define $\mathcal P_{\mathrm{inv}} := \left( \prod_{i=1}^{\mu}\prod_{j=1}^{g}a_{i,j}  \right) \left(  \prod_{t=1}^{\lambda}\prod_{j=1}^{g}b_{t,j} \right),$ and let $\mathcal P := \mathcal P_F  \mathcal P_I  \mathcal P_{\mathrm{inv}}.$ The polynomial \(\mathcal P\) is nonzero, and its total degree depends only on \(r,\ell,g,\alpha,\mu,\lambda\). Hence, for \(q>\deg\mathcal P\), the Schwartz--Zippel lemma~\cite{schwartz1980fast,zippel1979probabilistic} guarantees an assignment $a_{i,j},b_{t,j}\in\mathbb F_q$ such that \(\mathcal P\ne0\). Since \(\mathcal P_{\mathrm{inv}}\) is a factor of \(\mathcal P\), all \(a_{i,j}\) and \(b_{t,j}\) are nonzero under this assignment. Consequently, $A_i=D(a_{i,1},\ldots,a_{i,g}), i\in[\mu],$ and $B_t=D(b_{t,1},\ldots,b_{t,g}),  t\in[\lambda],$ belong to \(\mathcal D_g\); in particular, \(B_tA_i\in\mathcal D_g\) for every \(t\in[\lambda]\) and \(i\in[\mu]\).

    Under this assignment, $\mathcal P_F(A_1,\ldots,A_\mu)\ne0.$ Hence $\Delta_\omega(A_1,\ldots,A_\mu)\ne0  ~\text{for every }\omega\in\Omega_\mu^+.$ By the criterion proved above, $\mathcal L_{\mathcal C}(A_1,\ldots,A_\mu)$ is an optimal-distance \((\mu r,g,r,\ell,\alpha)\)-LRC. Likewise, under the same assignment, $\mathcal P_I (A_1,\ldots,A_\mu,B_1,\ldots,B_\lambda)$ $ \ne0.$ Therefore,$ \Delta_\omega^I (A_1,\ldots,A_\mu,B_1,\ldots,B_\lambda) \ne0 ~\text{for every }\omega\in\Omega_{\mathrm I}^+.$ Applying the criterion with the local groups indexed by \((t,i)\in[\lambda]\times[\mu]\), we conclude that $\mathcal L_{\mathcal C} \bigl( B_tA_i: t\in[\lambda],\ i\in[\mu] \bigr)$ is an optimal-distance \((\lambda\mu r,g,r,\ell,\alpha)\)-LRC. This completes the proof.
\end{proof}


\bibliographystyle{IEEEtran}
\bibliography{ref}

\end{document}